%% file: main.tex
\documentclass[11pt]{article}
\usepackage[T1]{fontenc}
\usepackage{lmodern}
\usepackage{geometry} 
\usepackage{graphicx}
\usepackage{amsmath,amssymb,amsthm,amsfonts,dsfont,mathtools}
\usepackage{enumitem}
\usepackage{relsize}
\usepackage[margin=1cm]{caption}
\usepackage{wrapfig}
\usepackage{frame,color}
\usepackage{environ,subfig}
\usepackage{hyperref}
\usepackage{xspace}
\usepackage{framed}
\usepackage{comment}
\usepackage{changepage}
\usepackage[linesnumbered,boxed]{algorithm2e}
\usepackage{algorithmic}
\usepackage{mathrsfs}
\usepackage{wasysym}
\usepackage{tabu} 
\usepackage[sort,nocompress]{cite}

\input{commands.tex}

\title{The Complexity Landscape of Distributed \\ Locally Checkable Problems on Trees\footnote{We thank Lickan Dang, Sebastian Brandt, and Chiara Piombi for pointing out errors in an earlier version of Lemma~\ref{lem-decomp-2}.}}
\author{Yi-Jun Chang\footnote{National University of Singapore. ORCID: 0000-0002-0109-2432. Email: cyijun@nus.edu.sg}}

\begin{document}
\date{}
\maketitle 

\begin{abstract}
Recent research revealed the existence of \emph{gaps} in the  complexity landscape of  {\em locally checkable labeling} (LCL) problems in the $\LOCAL$ model of distributed computing.  For example, the deterministic round complexity of any LCL problem on bounded-degree graphs is  either $O(\log^\ast n)$ or $\Omega(\log n)$ {[}Chang, Kopelowitz, and Pettie, FOCS 2016{]}. The complexity landscape of LCL problems is now quite well-understood, but a few questions remain open. 

For bounded-degree trees, there is an LCL problem with round complexity $\Theta(n^{1/k})$ for each positive integer $k$ {[}Chang and Pettie, FOCS 2017{]}. It is conjectured that 
no LCL problem has round complexity $o(n^{1/(k-1)})$ and $\omega(n^{1/k})$ on bounded-degree trees. As of now, only the case of $k = 2$ has been proved {[}Balliu et al., DISC 2018{]}.

In this paper, we show that for  LCL problems on bounded-degree trees, there is indeed a gap  between   $\Theta(n^{1/(k-1)})$ and $\Theta(n^{1/k})$ for each $k \geq 2$. Our proof is \emph{constructive} in the sense that it offers a  sequential algorithm that decides which side of the gap a given   LCL problem belongs to.
We also show that 
 it is $\EXPTIME$-hard to distinguish between $\Theta(1)$-round and $\Theta(n)$-round LCL problems on bounded-degree trees.  This improves upon a previous $\PSPACE$-hardness result {[}Balliu et al., PODC 2019{]}.
\end{abstract}

\thispagestyle{empty}
\newpage
\thispagestyle{empty}
\tableofcontents
\newpage
\pagenumbering{arabic}

\input{intro.tex}
\input{LCLgapBasic.tex}
\input{treeDecomp.tex}

\input{LCLgapExtend.tex}
\input{hardness.tex}

\newpage

\appendix
\section*{\LARGE Appendix}
\vspace{0.5cm}
\input{operations.tex}
\input{LCLgap.tex}

\newpage

\bibliographystyle{plain}
\bibliography{reference}

\end{document}

%% file: commands.tex
\newtheorem{theorem}{Theorem}
\newtheorem{lemma}{Lemma}

\newtheorem{definition}{Definition}

\newcommand{\TLLL}{T_{\operatorname{LLL}}}

\newcommand{\acc}{\mathsf{accept}}
\newcommand{\rej}{\mathsf{reject}}
\newcommand{\err}{\mathsf{error}}

\newcommand{\ignore}[1]{}

\newcommand{\floor}[1]{\lfloor #1 \rfloor}

\newcommand{\dist}{\operatorname{dist}}

\newcommand{\ID}{\operatorname{ID}}
\newcommand{\LOCAL}{\mathsf{LOCAL}}

\newcommand{\Lovasz}{Lov\'{a}sz}

\newcommand{\Pset}{\mathscr{P}}
\newcommand{\Hset}{\mathscr{H}}
\newcommand{\Tset}{\mathscr{T}}
\newcommand{\Sset}{\mathscr{S}}

\newcommand{\GC}[1]{{G}_{#1}^{\sf C}}
\newcommand{\GR}[1]{{G}_{#1}^{\sf R}}
\newcommand{\VC}[1]{{V}_{#1}^{\sf C}}
\newcommand{\VR}[1]{{V}_{#1}^{\sf R}}
\newcommand{\RRC}[1]{{\mathcal{R}}_{#1}^{\sf C}}
\newcommand{\RRR}[1]{{\mathcal{R}}_{#1}^{\sf R}}
\newcommand{\TTC}[1]{{\mathcal{T}}_{#1}^{\sf C}}
\newcommand{\TTR}[1]{{\mathcal{T}}_{#1}^{\sf R}}

\newcommand{\rake}{{\sf Rake}}
\newcommand{\compress}{{\sf Compress}}
\newcommand{\sk}{{\sf skel}}
\newcommand{\vi}{{\sf virt}}

\newcommand{\type}{\textsf{Type}}
\newcommand{\class}{\textsf{Class}}
\newcommand{\replace}{\textsf{Replace}}
\newcommand{\extend}{\textsf{Extend}}
\newcommand{\labelling}{\textsf{Label}}
\newcommand{\ext}{\textsf{Label}\mbox{-}\textsf{Extend}}
\newcommand{\pump}{\textsf{Pump}}
\newcommand{\truncate}{\textsf{Change}}
\newcommand{\cut}{\textsf{Duplicate}\mbox{-}\textsf{Cut}}
\newcommand{\bottom}{\perp}
\newcommand{\LabelIn}{\Sigma_{\operatorname{in}}}
\newcommand{\LabelOut}{\Sigma_{\operatorname{out}}}
\newcommand{\Lpump}{ {\ell_{\operatorname{pump}}} }

\newcommand{\simm}{{\overset{\star}{\sim}}}

\newcommand{\GG}{\mathcal{G}}
\newcommand{\HH}{\mathcal{H}}

\newcommand{\TT}{\mathcal{T}}
\newcommand{\LL}{\mathcal{L}}
\newcommand{\XX}{\mathcal{X}}
\newcommand{\YY}{\mathcal{Y}}
\newcommand{\ZZ}{\mathcal{Z}}

\newcommand{\MM}{\mathcal{M}}

\newcommand{\EXPTIME}{\mathsf{EXPTIME}}
\newcommand{\PSPACE}{\mathsf{PSPACE}}

%% file: intro.tex
\section{Introduction}

In this paper, we consider Linial's $\LOCAL$ model of distributed computing~\cite{Linial92,Peleg00}, where the input graph $G=(V,E)$ and the communication network are identical. Each vertex $v \in V$ corresponds to a processor, each edge $e \in E$ corresponds to a communication link, and the computation proceeds in
synchronized rounds. There is no restriction on the local computation power and the message size. The main complexity measure for an algorithm is the number of rounds.
We assume that the number of vertices $n = |V|$ and the maximum degree $\Delta = \max_{v \in V} \deg(v)$ are global knowledge.

There is a recent line of research~\cite{BrandtEtal16,balliu2019hardness,ChangKP16,ChangP17,FischerG17,GhaffariHK18,Balliu18,ChangHLPU18,Balliu2018disc,BalliuBOS19randomness,RozhonG19arxiv} aiming to systematically understand the round complexity of distributed graph problems, with a focus on the {\em locally checkable labelings} (LCL) problems~\cite{NaorS95}, which is the class of distributed problems whose solution is locally verifiable by examining a constant-radius neighborhood of each vertex.
The class of LCL problems is sufficiently general that it encompasses many well-studied problems  in the $\LOCAL$ model, such as maximal matching, maximal independent set, $(\Delta+1)$-vertex coloring, and sinkless orientation.

For example, in the $(\Delta+1)$-vertex coloring problem, the output of each vertex $v$ is a color $c(v) \in \{1, 2, \ldots, \Delta + 1\}$. The output is a \emph{legal solution} if $c(u) \neq c(v)$ for each edge $e = \{u, v\} \in E$.  Each vertex $v$ can locally check if it has a neighbor $u \in N(v)$ with $c(u) = c(v)$ by examining the output within its radius-1 neighborhood.

By restricting ourselves to LCLs, we can avoid dealing with
  uninteresting artificial problems, such as the problem that asks each vertex to gather the IDs of all vertices within radius $\sqrt{n}$.

\subsection{The Spectrum of Distributed Complexities}

Different from the sequential setting such as the Turing machine or the RAM model, in the complexity landscape of LCL problems in the $\LOCAL$ model,   several large \emph{gaps} exist in the complexity landscape.

\begin{figure}[h!]
\centering
{\includegraphics[width=1\textwidth]{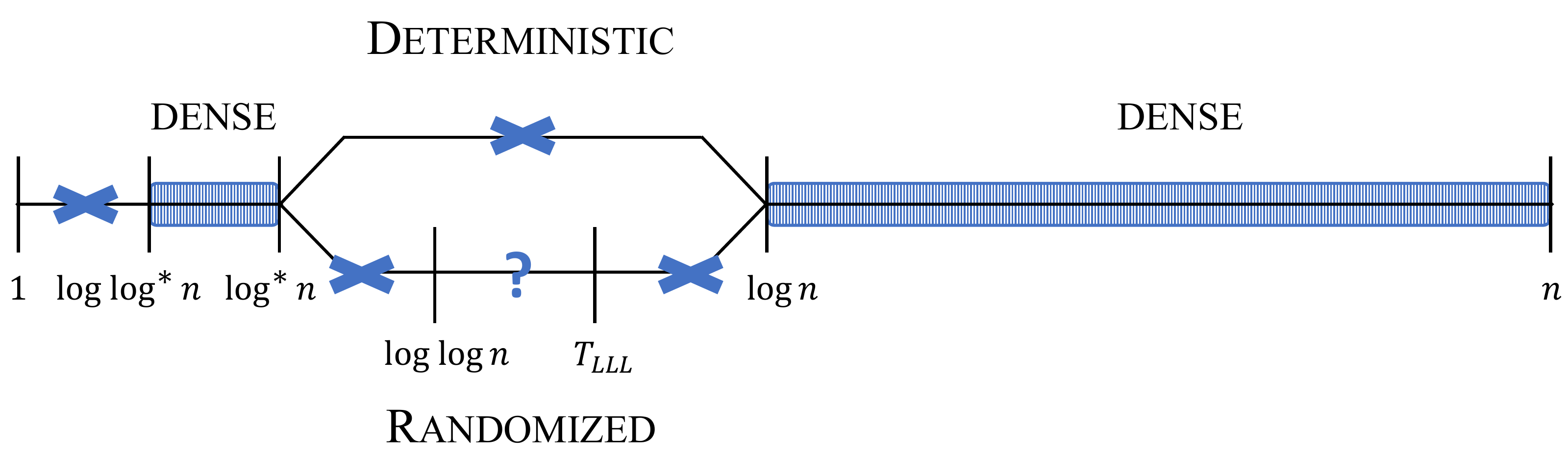}}
\caption{\label{fig:LCL-landscape-general} Complexity landscape of LCLs on bounded-degree general graphs}
\end{figure}

\paragraph{General graphs.}
Chang, Kopelowitz, and Pettie~\cite{ChangKP16} showed that for any LCL problem on bounded-degree graphs, its deterministic round complexity is either $O(\log^\ast n)$ or $\Omega(\log n)$, and its randomized round complexity is either $O(\log^\ast n)$ or $\Omega( \log \log n)$. Chang and Pettie~\cite{ChangP17} showed that any $o(\log n)$-round randomized algorithm for an LCL problem can be accelerated to run in $O(\TLLL)$ rounds, where $\TLLL$ is the randomized  complexity of the  {distributed constructive \Lovasz\ Local Lemma} (LLL)~\cite{ChungPS17} under a polynomial criterion $p d^{c} = O(1)$ for any positive constant $c$. It was conjectured in~\cite{ChangP17} that  $\TLLL = \Theta(\log \log n)$.
Chang and Pettie also showed that the gap $\omega(1)$--$o(\log \log^\ast n)$ can be derived using the approach of Naor and Stockmeyer~\cite{NaorS95} which is based on Ramsey theory.
Balliu~et~al.~showed that the two remaining regions $[\Theta(\log \log^\ast n), \Theta(\log^\ast n)]$ and $[\Theta(\log  n), \Theta(n)]$ are \emph{dense} in that many round complexity functions within these ranges can be realized by LCL problems~\cite{Balliu18}. See Figure~\ref{fig:LCL-landscape-general} for an illustration of the complexity landscape of LCLs on bounded-degree general graphs. 

\begin{figure}[h!]
\centering
{\includegraphics[width=1\textwidth]{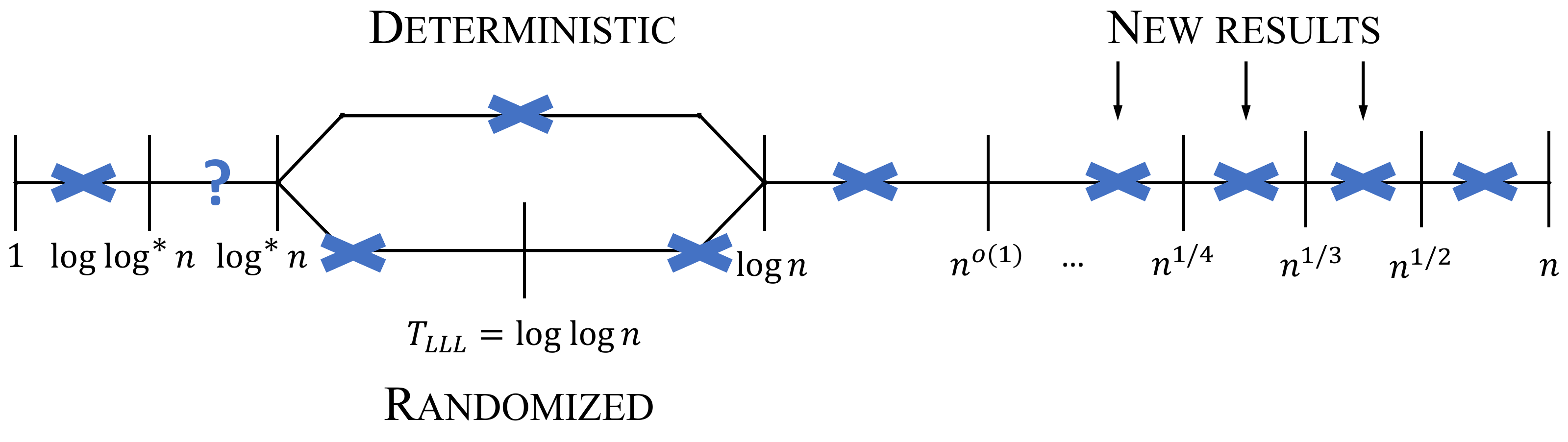}}
\caption{\label{fig:LCL-landscape-tree} Complexity landscape of LCLs on bounded-degree trees}
\end{figure}

\paragraph{Trees.} For bounded-degree trees, the four gaps in the lower end of the spectrum $[\Theta(1), \Theta(\log n)]$ are the same as that of general graphs. Chang~et~al.~\cite{ChangHLPU18} showed that $\TLLL = O(\log \log n)$ on bounded-degree trees.
It was conjectured in~\cite{ChangP17} that the $\omega(1)$--$o(\log \log^\ast n)$ gap on bounded-degree trees can be extended to  $\omega(1)$--$o(\log^\ast n)$. So far this conjecture was proved only for the special case of \emph{homogeneous} problems~\cite{balliu2019hardness}.
For the higher end of the spectrum  $[\Theta(\log n), \Theta(n)]$, Chang and Pettie~\cite{ChangP17} showed that any distributed algorithm that takes $n^{o(1)}$ rounds on bounded-degree trees can be accelerated to run in just $O(\log n)$ rounds, and there exists an LCL problem with complexity $\Theta(n^{1/k})$ for each $k \geq 1$. It was left as an open problem to decide if there are gaps between them.  Recently, Balliu et al.~\cite{Balliu2018disc} showed that there is indeed a gap between $\Theta(\sqrt{n})$ and $\Theta(n)$, but the other cases are still open. 
See Figure~\ref{fig:LCL-landscape-tree} for an illustration of the complexity landscape of LCLs on bounded-degree trees.

\paragraph{New result.} We prove the existence of the gap $\omega(n^{1/k})$--$o(n^{1/(k-1)})$ of LCL problems on bounded-degree trees, for each  integer $k \geq 2$. This result is obtained by unifying the approach of Balliu et al.~\cite{Balliu2018disc} and the approach of Chang and Pettie~\cite{ChangP17} using a generalized tree decomposition algorithm. 

\begin{theorem}\label{thm-gap-main}
For any integer $k \geq 2$, for any LCL problem $\mathcal{P}$ on bounded-degree trees, either one of the following holds.
\begin{itemize}
    \item The deterministic and randomized round complexities of $\mathcal{P}$ are $\Omega(n^{1/(k-1)})$.
    \item The deterministic and randomized round complexities of $\mathcal{P}$ are $O(n^{1/k})$.
\end{itemize}
Furthermore, there is a sequential algorithm that given an integer $k \geq 2$ and a description of $\mathcal{P}$ decides which side of the gap $\omega(n^{1/k})$--$o(n^{1/(k-1)})$ the problem $\mathcal{P}$  belongs to.
\end{theorem}

Theorem~\ref{thm-gap-main} implies that randomness does not help for LCL problems on bounded-degree trees in the regime of polynomial round complexity.

\subsection{The Complexity of Classification}

The complexity gaps in Figures~\ref{fig:LCL-landscape-general} and~\ref{fig:LCL-landscape-tree} classify the distributed problems into complexity classes.
A natural question to ask is whether this classification is \emph{decidable}.
Unfortunately, even for grids and tori, it is \emph{undecidable} whether a given LCL problem can be solved in $O(1)$ rounds~\cite{NaorS95,Brandt17}, as 
LCL problems on grids are expressive enough to simulate   Turing machines.
This undecidability result does not apply to special graph classes such as paths, cycles, and trees.
In fact, the proof of the $\omega(\log n)$--$n^{o(1)}$ gap on bounded-degree trees by Chang and Pettie~\cite{ChangP17} is \emph{constructive} in the sense that it offers an algorithm that  decides whether a given LCL problem on bounded-degree trees has complexity $O(\log n)$ or $n^{\Omega(1)}$.

Much progress has recently been made in  understanding to what extent the design of distributed algorithms and the proof of distributed lower bounds can be automated~\cite{Brandt17,Balliu2019a,Balliu2019decidable,Brandt19speedup,Chang2020DistributedGP,Olivetti20}.  
On paths or cycles, with or without input labels,   only three complexity classes are possible: $\Theta(1)$, $\Theta(\log^\ast n)$, and $\Theta(n)$.  Balliu et al.~\cite{Balliu2019decidable} showed that for any given LCL problem $\mathcal{P}$ on paths or cycles, it is \emph{decidable} to check which class $\mathcal{P}$ belongs to, and there is a sequential algorithm that automates the design of an asymptotically optimal distributed algorithm for $\mathcal{P}$.
For comparison, the previous proofs~\cite{ChangKP16,ChangP17,NaorS95} establishing this classification did not offer such results.

 On the negative side, Balliu et al.~\cite{Balliu2019decidable} showed that the problem of determining the optimal asymptotic distributed complexity of an LCL problem is $\PSPACE$-hard, even for paths and cycles with input labels. Since trees can be used to encode input labels, the same $\PSPACE$-hardness result extends to the case of bounded-degree trees without input labels.

\paragraph{New result.}
Our proof of the existence of the gap $\omega(n^{1/k})$--$o(n^{1/(k-1)})$  offers a  sequential algorithm that decides which side of the gap   a given   LCL problem $\mathcal{P}$ belongs to. When the locality radius $r$ of the LCL is a constant independent of the  description length $N$ of the LCL,   the runtime $2^{2^{N^{O(1)}}}$ of our sequential algorithm is \emph{doubly exponential} in $N^{O(1)}$.
To complement this result, we show that this problem is inherently very hard  by proving that 
 this problem is $\EXPTIME$-hard.
Specifically, we say that a round complexity function $T(n)$ is \emph{realizable} if there exists an LCL problem $\mathcal{P}$ whose  round complexity is $\Theta(T(n))$ on bounded-degree trees. We prove the following theorem.

\begin{theorem}\label{thm-hardness}
Let $T_1(n) \ll T_2(n)$ be two realizable round complexity functions. Given an LCL problem $\mathcal{P}$ that is promised to have round complexity either $\Theta(T_1(n))$ or $\Theta(T_2(n))$ on bounded-degree trees, it is $\EXPTIME$-hard to decide the round complexity of $\mathcal{P}$.
\end{theorem}

\subsection{Organization}
In Section~\ref{sect-prelim}, we  overview  the basics of  LCL problems and review the pumping lemma of Chang and Pettie~\cite{ChangP17}.  In Section~\ref{sect-gap-review},  we review the proof of the $\omega(n^{1/2})$--$o(n)$ gap by Balliu et al.~\cite{Balliu2018disc}.
In Section~\ref{sec:decomp}, we consider a generalized version of the tree decomposition of Miller and Reif~\cite{MillerR89} that allows us to unify the approach of Balliu et al.~\cite{Balliu2018disc} and the approach of Chang and Pettie~\cite{ChangP17}. In Section~\ref{sect:gap-extend}, we prove Theorem~\ref{thm-gap-main} for the case of deterministic algorithms, and the complete proof of Theorem~\ref{thm-gap-main} is left to the appendix. 
In Section~\ref{sect:hardness}, we prove Theorem~\ref{thm-hardness}.

\section{Preliminaries}\label{sect-prelim}

In the deterministic variant of the  $\LOCAL$ model, each vertex $v$  has a distinct $O(\log n)$-bit identifier $\ID(v)$.
In the randomized variant of the $\LOCAL$ model, there are no distinct identifiers, but each vertex  has access to a stream of unbiased random bits, and the maximum tolerable global probability of failure is $1/n$.
Note that a $t$-round $\LOCAL$ algorithm can be seen as a function that maps a radius-$t$ subgraph  centered at $v$ to an output label assigned to $v$.

\subsection{Locally Checkable Labeling} A distributed graph problem is locally checkable if there is some constant $r$ such that the validity of a solution can be checked locally by having each vertex examine its radius-$r$ neighborhood. For example, the maximal independent set  problem is locally checkable with locality radius $r=1$, but the maximum independent set problem is not locally checkable.

\paragraph{Formal definition.} Formally, an LCL problem $\mathcal{P}$ is specified by the following parameters: the locality radius $r$, the set of input labels $\LabelIn$, the set of output labels  $\LabelOut$, and the set of allowed configurations $\mathcal{C}$. Each member of $\mathcal{C}$ is a radius-$r$ subgraph $H$ centered at a specific vertex $v$, where each vertex  in $H$ is assigned an input label from $\LabelIn$ and an output label from $\LabelOut$. Note that $|\LabelIn| = 1$ corresponds to the special case where there is no input label.

 An \emph{instance} of an LCL problem $\mathcal{P}$ is a graph $G=(V,E)$ where each vertex is assigned an input label from  $\LabelIn$. A \emph{solution} for  $\mathcal{P}$  on $G$ is a labeling function $\phi_{\text{out}}$ that assigns to each vertex in $G$ an output label from  $\LabelOut$. 
We say that $\phi_{\text{out}}$ is \emph{locally consistent} for a vertex $v \in V$ if  its radius-$r$ neighborhood $N^r(v)$ is an allowed configuration in $\mathcal{C}$ under the given input labeling and the output labeling $\phi_{\text{out}}$. The output labeling $\phi_{\text{out}}$ is \emph{legal} if it is locally consistent everywhere.

\paragraph{Graph terminology.} Unless otherwise stated, all vertices in all graphs in this paper are assigned input labels from $\LabelIn$, and the term \emph{label} refers to \emph{output label}. An \emph{unlabeled graph} is a graph where no vertex is assigned an output label from $\LabelOut$. A \emph{partially labeled graph} is a graph with a labeling function $\LL$ that maps each vertex $v$ to an element of $\LabelOut \cup \{\bot\}$. A \emph{completely labeled graph} 
is a graph with a labeling function $\LL$ that maps each vertex $v$ to an element of $\LabelOut$.

\paragraph{Description length.} We assume that any given LCL problem $\mathcal{P}$ is specified by representing  $\mathcal{C}$ as a truth table.
Specifically, a \emph{centered graph} is a graph $G=(V,E)$ with a distinguished vertex $s \in V$, and the \emph{radius} of $G$ is defined by $\max_{v \in V} \dist(v,s)$.  
The truth table representation of $\mathcal{P}$ is a mapping \[\mathscr{G}_{r,\Delta,\LabelIn,\LabelOut} \mapsto \{0, 1\},\] where $\mathscr{G}_{r,\Delta,\LabelIn,\LabelOut}$ is the set of all  centered graphs $G=(V,E)$ of radius at most $r$ with maximum degree $\Delta$ where each vertex $v \in V$ is equipped with an input label from $\LabelIn$ and an output label from $\LabelOut$.

If the graph class under consideration is the set of trees of maximum degree $\Delta$, then the description length of $\mathcal{P}$ can be upper bounded by  \[\left(1 + |\LabelIn| \cdot |\LabelOut|\right)^{1 + \Delta^r}.\]

To derive this upper bound, consider the rooted tree $T_r$ of height $r$ where the root $v$ has $\Delta$ children, all vertices $u$ with $1 \leq \dist(u,v) \leq r-1$ have $\Delta-1$ children, and all vertices  $u$ with $\dist(u,v) = r$ are leaf vertices.
The number of trees in $\mathscr{G}_{r,\Delta,\LabelIn,\LabelOut}$ is at most the number of distinct labeling of the vertices in $T_r$ by $(\LabelIn \times \LabelOut) \cup \{\fullmoon\}$, where $\fullmoon$ is a special symbol indicating the non-existence of a vertex.
Hence the description length can be upper bounded by $\left(1 + |\LabelIn| \cdot |\LabelOut|\right)^{n_r}$, where $n_r$ is the number of vertices in $T_r$. We have  $n_0 = 1$, $n_1 = 1+ \Delta$, and $n_r = 1 + \Delta + \Delta \sum_{i=1}^{r-1} (\Delta-1)^{r-1}$ for each $r \geq 2$. It is clear that $n_r \leq 1 + \Delta^r$ for all $r$. 

\paragraph{Remarks on edge labeling and orientation.} In general, an LCL might have edge labels and edge orientations. It is straightforward to encode edge labels and edge orientations as vertex labels. For example, given an input graph $G$, consider the following pre-processing. For each edge $e = \{u,v\} \in E$, subdivide it into a length-3 path $(u, x_{e,u}, x_{e,v}, v)$ by adding two new vertices $x_{e,u}$ and $x_{e,v}$. Each  newly added vertex is assigned a special input label $\sf{e}$ indicating that it represents a half of an edge. 
Now an edge orientation $u \rightarrow v$ can be encoded as $\phi(x_{e,u}) = 0$ and $\phi(x_{e,v}) = 1$.

\subsection{Pumping Lemma}\label{sect:pumping-prelim}

We review the pumping lemma of Chang and Pettie~\cite{ChangP17}, which plays a crucial role in establishing complexity gaps on trees. 

\paragraph{Notation for partially labeled graphs.}
A \emph{partially labeled graph} $\GG = (G,\LL)$ is a graph $G=(V,E)$ together with a function
$\LL : V \rightarrow \LabelOut \cup \{\bottom\}$.  The vertices in $\LL^{-1}(\bottom)$ are \emph{unlabeled}.
A \emph{complete labeling} $\LL' : V(G)\rightarrow \LabelOut$ for $\GG$
is one that labels all vertices and is consistent with the partial labeling of $\GG$, i.e., $\LL'(v) = \LL(v)$ whenever $\LL(v)\neq\: \bottom$.
A \emph{legal labeling} is a complete labeling that is \emph{locally consistent} for all $v\in V(G)$,
i.e., the labeled subgraph induced by $N^r(v)$ is consistent with the given LCL problem $\mathcal{P}$.
Here $N^r(v)$ is the set of all vertices within distance $r$ of $v$.
A subgraph of a partially labeled graph $\GG=(G,\LL)$ is a pair $\HH=(H,\LL')$
such that $H$ is a subgraph of $G$, and $\LL'$ is $\LL$ restricted to the domain $V(H)$.
With a slight abuse of notation, we usually write $\HH=(H,\LL)$.

\paragraph{An equivalence relation.} A tree $\HH$ with two distinguished vertices $s, t \in V(H)$ is called a \emph{bipolar} tree. We call $s$ and $t$  the two \emph{poles} of $\HH$. We consider the equivalence relation  $\simm$ on bipolar trees defined in~\cite{ChangP17}. We write $\type(\HH)$ to denote the equivalence class of the bipolar tree $\HH$. The exact definition of  $\simm$ is not important. What is crucial is the following property of   $\simm$. 

\begin{framed} 
\noindent Suppose we are given the following.
\begin{itemize}
    \item $\GG$ is a graph.
    \item $\HH$ is a bipolar subtree of $\GG$ with two poles $s$ and $t$ such that the removal of $s$ and $t$ disconnects $\HH$ from the rest of $\GG$. 
    \item $\HH'$ is another bipolar subtree  with two poles $s'$ and $t'$ such that $\type(\HH) = \type(\HH')$.
    \item $\LL_\diamond$ is a complete legal labeling of $\GG$. 
\end{itemize} 
Define the graph $\GG'$ as the result of
replacing the subgraph $\HH$ of $\GG$ with  $\HH'$. Then there exists a legal labeling $\LL'$ of $\HH'$ meeting the following conditions.

\begin{itemize}
    \item 
The following complete labeling $\LL_\diamond'$ of $\GG'$ is a legal labeling.
\[
\LL_\diamond'(v) =
\begin{cases}
\LL'(v)  &\text{if $v \in \HH$,}\\
\LL_\diamond(v)  &\text{if $v \in \GG \setminus \HH$.}
\end{cases}\]
\item  
Such a labeling $\LL'$ of $\HH'$ can be computed \emph{solely} from $\HH'$ and the given labeling $\LL_\diamond$ restricted to $\HH$.
\end{itemize}
\end{framed}

In view of the above, the vertices in $\HH'$ can compute their $\LL'$-labels using only information within $\HH'$ and the given labeling $\LL_\diamond$ restricted to $\HH$, without communicating with the vertices outside of $\HH'$.
Intuitively, this allows us to reduce the task of finding a legal labeling $\LL_\diamond'$ of $\GG'$ to  the task of finding a legal labeling $\LL_\diamond$ of $\GG$. 

\paragraph{A pumping lemma for bipolar trees.}
 The unique path $(s = u_1, u_2, \ldots, u_k = t)$   connecting the two poles $s$ and $t$ of a bipolar tree $\HH$ is called the \emph{core path} of $\HH$.
 The tree $\HH$ can be viewed as a string of subtrees $\TT_1, \TT_2, \ldots, \TT_k$, where $\TT_i$ is the subtree of $\HH$ rooted at $u_i$. 
 For  convenience, we use the   string notation  $\HH=(\TT_1, \TT_2, \ldots, \TT_k)$ to describe a bipolar tree $\HH$. 
 Viewing bipolar trees as strings, the following \emph{pumping lemma} was proved in~\cite{ChangP17}.

 \begin{framed}
\noindent  There exists a number $\Lpump$ depending only on the given LCL problem $\mathcal{P}$ such that as long as $k \geq \Lpump$, any bipolar tree $\HH =(\TT_1, \TT_2, \ldots, \TT_k)$ can be decomposed into three substrings $\HH = x \circ y \circ z$  meeting the following conditions.
\begin{itemize}
    \item $|xy| \leq \Lpump$.
    \item $|y|\geq 1$.
    \item  $\type(x \circ y^j \circ z) = \type(\HH)$ for each non-negative integer $j$.
\end{itemize}
 \end{framed}

 Intuitively, the pumping lemma allows us to extend the length of $\HH =(\TT_1, \TT_2, \ldots, \TT_k)$ to arbitrarily long without changing its type, as long as $k \geq \Lpump$. 

%% file: LCLgapBasic.tex
\section{A Review of the \texorpdfstring{$\omega(n^{1/2})$--$o(n)$}{omega(sqrt(n))--o(n)} Gap}\label{sect-gap-review}
We briefly review the proof of the $\omega(n^{1/2})$--$o(n)$ gap by Balliu et al.~\cite{Balliu2018disc}.
Given an $o(n)$-round randomized or deterministic $\LOCAL$ algorithm $\mathcal{A}$ for the given LCL problem $\mathcal{P}$, the goal is to design a new randomized or deterministic $\LOCAL$ algorithm $\mathcal{A}'$ with round complexity $O(\sqrt{n})$.

Within this section, we only apply the pumping lemma on unlabeled graphs, but we will see that when we extend the proof to other gaps, we need to deal with partially labeled graphs.

\paragraph{The skeleton tree.} Let the tree $G=(V,E)$ be the underlying network.
Let $\tau = \Theta(\sqrt{n})$ be a threshold to be determined. Define the {\em skeleton tree} $G_{\sk}$ as the result of iteratively removing all leaf vertices of $G$ for $\tau$ iterations. Specifically,  start with $G_0 = G$,  and let $G_i$ be the result of removing all leaf vertices of $G_{i-1}$ for each $1 \leq i \leq \tau$, and then we have $G_{\sk} = G_{\tau}$.

If $G_{\sk}$ is empty, then we are already done, since this implies that the diameter of $G$ is $O(\tau) = O(\sqrt{n})$, so $\mathcal{P}$ can be solved trivially in $O(\sqrt{n})$ rounds. In the subsequent discussion we assume that $G_{\sk}$ is not empty. We will identify a set of disjoint paths $\Pset$ of $G_{\sk}$ meeting the following conditions.
\begin{enumerate}
    \item \label{cc1} Each  path $P=(v_1, v_2, \ldots, v_x) \in \Pset$ satisfies the following requirements. 
    \begin{enumerate}
        \item $x \in [\Lpump, 2\Lpump]$.
        \item Each $v_i$ is of degree-2 in $G_{\sk}$.
    \end{enumerate}
    \item \label{cc2} Let $G'$ be the subgraph of $G_{\sk}$ resulting from removing all paths in $\Pset$. Let $\Sset$ denote the set of connected components in $G'$. Then each connected component $S \in \Sset$ in $G_{\sk}$ has diameter $O(\sqrt{n})$.
\end{enumerate}

The proof of the existence of $\Pset$ can be found in~\cite{Balliu2018disc}. We will also provide a proof in Section~\ref{sec:decomp}.
Here we only need to use the fact that the skeleton tree $G_{\sk}$ and the set of paths  $\mathcal{P}$  can be computed in $O(\sqrt{n})$ rounds on $G$.

Since $G_{\sk}$ is constructed by iteratively removing all leaf vertices of $G$ for $\tau$ iterations, each vertex $v$ in $G \setminus G_{\sk}$ is reachable to a \emph{unique} vertex $u$ in $G_{\sk}$ via the vertices $G \setminus G_{\sk}$.

For any vertex subset  $U$ in  $G_{\sk}$, we define $U^\ast \supseteq U$ as the set of vertices in $G$ resulting from adding to $U$ all vertices in $G \setminus G_{\sk}$ reachable to $U$ via the vertices in $G \setminus G_{\sk}$. A crucial consequence of Condition~\ref{cc2} is that the diameter of  $S^\ast$ is $O(\tau + \sqrt{n}) = O(\sqrt{n})$, for each $S \in \Sset$.


\paragraph{The virtual tree.} Consider the {\em virtual tree} $G_{\vi}$ defined as the result of applying the pumping lemma on $P^\ast$ for each $P=(v_1, v_2, \ldots, v_x) \in \Pset$ to the graph $G$. The definition of $P^\ast$ is in the paragraph above.  Here  $P^\ast$ is seen as a bipolar tree with the poles $s=v_1$ and $t=v_x$. 
Specifically, the pumping lemma allows us to replace each bipolar tree $P^\ast=(T_1, T_2, \ldots, T_k)$ is by some other bipolar tree $P' = (T_1', T_2', \ldots, T_{x'}')$ such that $\type(P') = \type(P^\ast)$, and   $x' \in [w, w + \Lpump]$, where $w$ is some very large number to be determined. 

\paragraph{The $O(\sqrt{n})$-round algorithm $\mathcal{A}'$.} We are ready to describe the $O(\sqrt{n})$-round algorithm $\mathcal{A}'$. The first step of the algorithm is to compute the skeleton tree $G_{\sk}$ and the set of paths $\Pset$ in $O(\sqrt{n})$ rounds. 
After that, we can simulate the virtual tree  $G_{\vi}$ by having the vertices in each $P \in \Pset$ simulate the virtual bipolar tree $P'$ resulting from the pumping lemma. We compute a legal labeling $\LL_\vi$  of $G_{\vi}$ by a simulation of $\mathcal{A}$ on  $G_{\vi}$. We will later see that the simulation can also be done in $O(\sqrt{n})$ rounds. Finally, we will show that the labeling $\LL_\vi$ can be transformed into a legal labeling $\LL$ of $G$ using another $O(\sqrt{n})$ rounds.

\paragraph{Simulation of $\mathcal{A}$ on the virtual tree.} 
It is clear that the number of vertices in $G_{\vi}$ can be upper bounded by $O(n^2 w)$, since $|\Pset| \leq n$ and the number of vertices in each bipolar tree $P'$ is $O(nw)$. We simulate the given algorithm $\mathcal{A}$ on the virtual tree  $G_{\vi}$  assuming that the number of vertices is $n' = O(n^2 w)$.

Since the round complexity of $\mathcal{A}$ on an $n'$-vertex graph is $o(n')$, by selecting $w$ as a sufficiently large number depending on $n$, the round complexity of $\mathcal{A}$ can be made much smaller than $0.1 w$.
Therefore, to simulate $\mathcal{A}$ on  $G_{\vi}$, each vertex $v$ in $G_{\vi}$ only needs to gather all information within radius $0.1 w$ to $v$. We make the following observations.
\begin{itemize}
    \item For each $S \in \Sset$, the subgraph $S^\ast$  has diameter $O(\sqrt{n})$.
    \item For each  $P\in \Pset$, the number of vertices in the core path of the bipolar subtree  $P'$ is within $[w, w + \Lpump]$. 
\end{itemize}
By these facts, it is straightforward to see that each vertex $v$ in $G_{\vi}$ is able to gather all information within radius $0.1 w$ to $v$ in $O(\sqrt{n})$ rounds of communication in the underlying network $G$.
For example, if $v \in S^\ast$ for some $S \in \Sset$, then $v$ only need to learn the following.
\begin{itemize}
    \item The subgraph induced by the set $S^\ast$.
    \item The virtual bipolar tree $P'$, for each  path $P \in \Pset$ adjacent to $S$.
\end{itemize}
Remember that $P'$ can be computed from $P^\ast$.
Since the diameter of $S^\ast$ (for each $S \in \Sset$)   and the diameter of $P^\ast$ (for each $P \in \Pset$) are  $O(\sqrt{n})$,
 this information gathering can be done in $O(\sqrt{n})$ rounds in $G$.

\paragraph{Computing a legal labeling of $G$.} Suppose we have computed a legal labeling $\LL_\vi$  of $G_{\vi}$. We show how to use this legal labeling $\LL_\vi$ to obtain a legal labeling $\LL$ of $G$ in $O(\sqrt{n})$ rounds. 
For each $S \in \Sset$, the labeling of the vertices in $S^\ast$ is unchanged, i.e., $\LL(v) = \LL_{\vi}(v)$.
For each path $P \in \Pset$, the $\LL$-labels of the vertices in $P^\ast$ are computed as follows. 

Remember that $G_{\vi}$ is the result of replacing $P^\ast$ with $P'$, for each $P \in \Pset$, and the two bipolar trees $P'$ and $P^\ast$ have the same type. 
In view of the property of $\simm$ described in Section~\ref{sect:pumping-prelim}, there exists a labeling $\LL'$ of $P^\ast$  such that if we replace the bipolar subtree $P'$ (labeled with  $\LL_{\vi}$) by the bipolar subtree $P^\ast$ (labeled with  $\LL'$), the legality of the labeling of the underlying graph is maintained. 
Moreover, such a labeling $\LL'$ of $P^\ast$  can be computed from the  labeling $\LL_{\vi}$ restricted to $P'$, without using any information outside of $P'$.
Thus, we can carry out this procedure  in parallel for each $P \in \Pset$, and this takes   $O(\sqrt{n})$ rounds, since the diameter of $P^\ast$ is at most $2\tau + 2\Lpump - 1 = O(\sqrt{n})$ for each $P \in \Pset$.
After that, we obtain a desired legal labeling
 $\LL$ of $G$.

%% file: treeDecomp.tex
\section{A Generalized Tree Decomposition} \label{sec:decomp}

Miller and Reif~\cite{MillerR89} considered the following decomposition algorithm. Start with a tree $G=(V,E)$, and then remove the vertices in $V$ by repeatedly doing the following two operations alternately: \rake\ (removing all leaf vertices) and \compress\ (removing all degree-2 vertices). It is known that $O(\log n)$ iterations suffice to remove all vertices in the tree~\cite{MillerR89}.
Variants of this decomposition have turned out to be useful in the design of $\LOCAL$ algorithms~\cite{ChangP17,ChangHLPU18}.

In this section, we consider a generalized version of this decomposition, which allows us to show the   existence of $\Pset$ needed in Section~\ref{sect-gap-review}, and to extend the proof idea in Section~\ref{sect-gap-review} to other gaps.

We begin with a formal definition of our decomposition, which is parameterized by two integers $\ell \geq 1$ and $\gamma \geq 1$, and it decomposes the vertices in the tree $G$ into \[V = \VR{1} \cup \VC{1} \cup \VR{2} \cup \VC{2} \cup \VR{3} \cup \VC{3} \cup \cdots.\] 

Let $L$ denote the largest index $i$ such that $\VC{i} \cup \VR{i+1} \cup \VC{i+1} \cup \cdots$ is empty.
We define $\GC{i}$ as the subgraph induced by the vertices  $\left(\bigcup_{j=i+1}^{L} \VR{j}\right) \cup \left( \bigcup_{j=i}^{L-1} \VC{j}\right)$, which is the set of all vertices that are in $\VC{i}$ or higher layers. Similarly, we define $\GR{i}$ as the subgraph induced by the vertices  
$\left(\bigcup_{j=i}^L \VR{j}\right) \cup \left( \bigcup_{j=i}^{L-1} \VC{j}\right)$.

We require the sets $\VR{i}$ and $\VC{i}$ to satisfy some requirements. Each connected component of the subgraph induced by $\VR{i}$ must be a rooted tree with height at most $\gamma-1$, and only the root can possibly have neighbors in $\VC{i} \cup \VR{i+1} \cup \VC{i+1} \cup \cdots$.  Each connected component of the subgraph induced by $\VC{i}$ must be a path with  $x \in [\ell, 2\ell]$ vertices, and only the endpoints can  possibly have neighbors in $\VR{i+1} \cup \VC{i+1} \cup \VR{i+2} \cup \cdots$. 
See Figure~\ref{fig:treedecomp} for an illustration, where each triangle represents a rooted tree.  The precise requirements are as follows.

\begin{figure}[ht!]
\centering
{\includegraphics[width=1\textwidth]{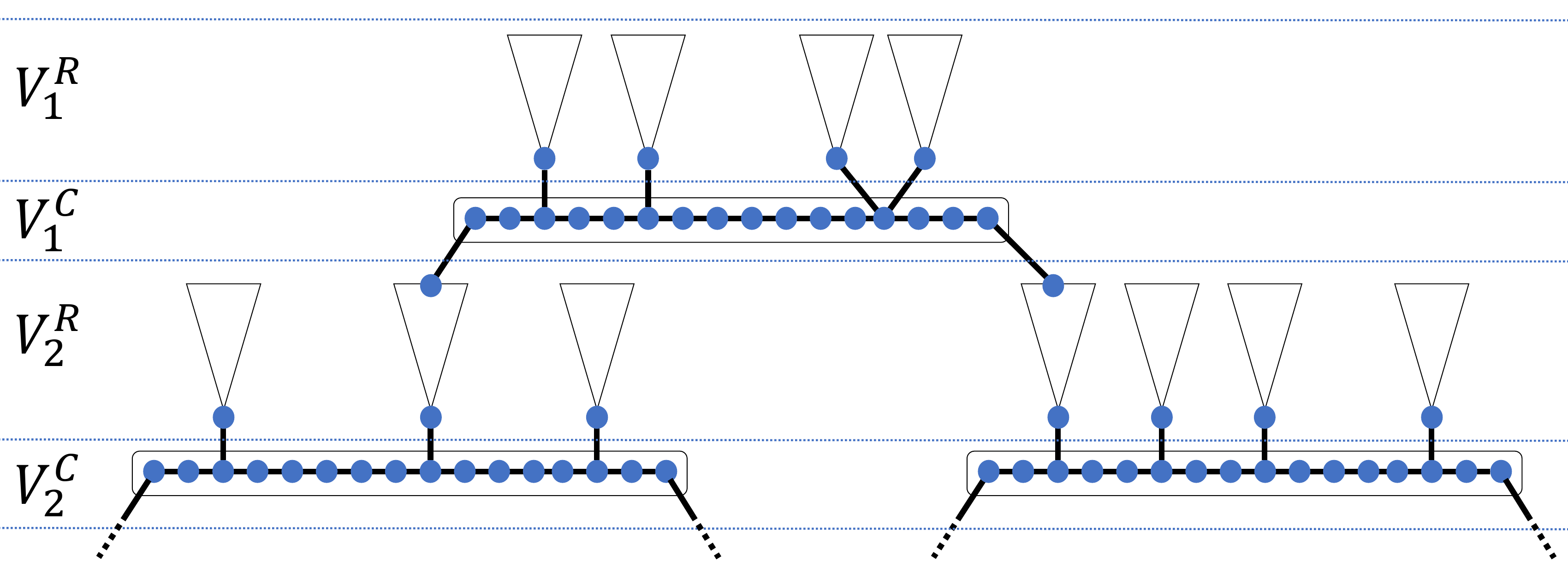}}
\caption{\label{fig:treedecomp} Top layers in a generalized tree decomposition}
\end{figure}

\paragraph{Requirements for $\VR{i}$.}
Let $S$ be a connected component of the subgraph induced by $\VR{i}$. Then there is a root vertex $z \in S$ such that the following conditions are met.
\begin{itemize}
    \item $z$ has at most one neighbor in $\GC{i}$, and each $v \in S \setminus \{z\}$ has no neighbor in $\GC{i}$.
    \item Each $v \in S \setminus \{z\}$ satisfies $\dist(v,z) \leq \gamma -1$.
\end{itemize}  
 Note that for the special case of $\gamma = 1$, the set $\VR{i}$ is an independent set.

\paragraph{Requirements for $\VC{i}$.}
Let $S$ be a connected component of the subgraph induced by $\VC{i}$. Then $S$ is a path $(u_1, u_2, \ldots, u_x)$ with $x \in [\ell, 2\ell]$ such that the following is true for each $u_j \in S$.
\begin{itemize}
    \item For the case $1 < j < x$ (i.e., $u_j$ is an intermediate vertex), $u_j$ has no neighbor in $\GR{i+1}$.
    \item Consider the case $j = 1$ or $j = x$ (i.e., $u_j$ is an endpoint). If $x \geq 2$, then $u_j$ has exactly one neighbor in $\GR{i+1}$. If $x = 1$, then $u_j$ has exactly two neighbors in $\GR{i+1}$. 
\end{itemize} 
Note that each $u_j \in S$ is of degree-2 in $\GC{i}$.

A decomposition $V = \VR{1} \cup \VC{1} \cup \VR{2} \cup \VC{2} \cup \VR{3} \cup \VC{3} \cup \cdots$ satisfying the above requirements is called a \emph{($\gamma, \ell$)-decomposition}.
We will see in  Lemma~\ref{lem-decomp-2} that for any positive integers $k = O(1)$  and $\ell = O(1)$, and for any parameter $\gamma$ satisfying \[ n^{1/k}(2\ell)^{1-1/k} \leq \gamma = O(n^{1/k}),\] an ($\gamma, \ell$)-decomposition with $L = k$  can be computed in $O(n^{1/k})$ rounds deterministically. 

\subsection{The Decomposition Algorithm}
Our algorithm constructing the above decomposition uses the following modified \rake\ and \compress\ operations defined in~\cite{ChangP17}. Here $U$ is a subset of $V$ representing the set of vertices that are not yet removed.
\begin{description}
    \item[Rake:]  Each  $v \in U$ removes itself if one of the following conditions is met.
    \begin{enumerate}
        \item \label{xx1} $\deg_U(v)  = 0$.
         \item  \label{xx2} $\deg_U(v)  = 1$ and the unique neighbor $u$ of $v$ in $U$ has $\deg_U(u) > 1$.
          \item   \label{xx3} $\deg_U(v)  = 1$ and the unique neighbor $u$ of $v$ in $U$ has $\deg_U(u) =1$ and $\ID(v)  > \ID(u)$. 
    \end{enumerate}
    \item[Compress:]  Each $v \in U$ removes itself if $v$ belongs to a path $P$ such that $|V(P)|\geq \ell$ and $\deg_U(u) = 2$ for each $u \in V(P)$.
\end{description}

The purpose of Condition~\ref{xx3} in the \rake\ operation is to break tie for the special case where $v$ is in a component of $U$ that is a length-1 path. This is to ensure that we remove an independent set of vertices in a \rake\ operation.

\begin{definition}[\cite{ChangP17}~]
Let $P$ be a path.  A subset $I\subset V(P)$ is called an $(\alpha,\beta)$-independent set
if the following conditions are met:
(i) $I$ is an independent set that does not contain either endpoint of $P$, and
(ii) each connected component of the subgraph induced by $V(P) \setminus I$ has at least $\alpha$ vertices and at most $\beta$ vertices,
unless $|V(P)|<\alpha$, in which case $I=\emptyset$.
\end{definition}

It is a folklore~\cite{Balliu2019decidable,ChangP17,Linial92} that an
{$(\ell,2\ell)$-independent set} of a path graph can be computed in $O(\log^* n)$ rounds deterministically when $\ell = O(1)$.

\paragraph{The algorithm.}
The decomposition algorithm begins with $U=V(G)$ and $i=1$.
In iteration $i$, we do the following.
\begin{enumerate}
    \item Do $\gamma$  \rake\ operations.
    \item Do one \compress\ operation.
    \item Update the iteration number $i \leftarrow i+1$.
\end{enumerate}
 We repeatedly do this until $U=\emptyset$, and then we proceed to the following post-processing step.

\paragraph{The post-processing step.}
Let $R_i$ (resp., $C_i$) be the set of vertices removed during a \rake\ (resp., \compress) operation in the $i$th iteration. For each path $P$ that is a connected component of the subgraph induced by $C_i$, Find an $(\ell,2\ell)$-independent set $I_P$ of $P$.
Define $C_i^\ast$ as the subset of $C_i$ that is the union of $I_P$ over all connected components  $P$ of the subgraph induced by $C_i$.

Let $L$ be the largest index $i$ such that $R_i \cup C_{i-1} \neq \emptyset$. Then a partition $V= \left(\bigcup_{i=1}^L \VR{i}\right) \cup \left( \bigcup_{i=1}^{L-1} \VC{i}\right)$ is defined by setting $\VR{i} = R_i \cup C_{i-1}^\ast$ and  $\VC{i} = C_i \setminus C_{i}^\ast$.

What we have done in the post-processing step is  promoting each vertex in the independent set  $I_P$ to the next layer, and this ensures that the requirement on the size of paths for $\VC{i}$ is met.

\paragraph{Analysis.}
We analyze the decomposition  $V= \left(\bigcup_{i=1}^L \VR{i}\right) \cup \left( \bigcup_{i=1}^{L-1} \VC{i}\right)$  produced using the above algorithm.
The proofs of the following two lemmas follow immediately from the description of the decomposition algorithm.

\begin{lemma}[Properties of $\VR{i}$] \label{lem:prop-of-vr}
Let $S$ be a connected component of the subgraph induced by $\VR{i}$. Then there is a root vertex $z \in S$ such that the following conditions are met.
\begin{itemize}
    \item $z$ has at most one neighbor in $\GC{i}$, and each $v \in S \setminus \{z\}$ has no neighbor in $\GC{i}$.
    \item Each $v \in S \setminus \{z\}$ satisfies $\dist(v,z) \leq \gamma -1$.
\end{itemize}  
\end{lemma}
\begin{proof}
The first case is when $S$ contains a vertex $u$ that is in $I_P$ for some $P$ during the post-processing step,  we must have $S = \{u\}$, and $u$ has no neighbor in $\GC{i}$. In this case, setting $z = u$ works.

The second case is when $S \subseteq R_i$. We select $z \in S$ as the last vertex removed from $U$ during the decomposition algorithm, among all vertices in $S$. 
Since we do $\gamma$ \rake\ operations in each iteration, each $v \in S \setminus \{z\}$ satisfies $\dist(v,z) \leq \gamma -1$. It is straightforward to see that $z$ is the only vertex in $S$ that may have a neighbor in $\GC{i}$; and $z$ can have at most one such neighbor.
\end{proof}


\begin{lemma}[Properties of $\VC{i}$] \label{lem:prop-of-vc}
Let $S$ be a connected component of the subgraph induced by $\VC{i}$. Then $S$ is a path $(u_1, u_2, \ldots, u_x)$ with $x \in [\ell, 2\ell]$ such that the following is true for each $u_j \in S$.
\begin{itemize}
    \item For the case $1 < j < x$ (i.e., $u_j$ is an intermediate vertex), $u_j$ has no neighbor in $\GR{i+1}$.
    \item Consider the case $j = 1$ or $j = x$ (i.e., $u_j$ is an endpoint). If $x \geq 2$, then $u_j$ has exactly one neighbor in $\GR{i+1}$. If $x = 1$, then $u_j$ has exactly two neighbors in $\GR{i+1}$. 
\end{itemize} 
\end{lemma}
\begin{proof}
In view of the post-processing step and the definition of an $(\alpha,\beta)$-independent set, $S$ is a path $(u_1, u_2, \ldots, u_x)$ with $x \in [\ell, 2\ell]$. It is straightforward to verify that the conditions specified in the lemma are met. 
\end{proof}

Next, we analyze the round complexity of the decomposition algorithm and the number $L$ in the decomposition. We remark that the  case of $\gamma = 1$ is considered and analyzed in~\cite{ChangP17}.

\begin{lemma}[\cite{ChangP17}~] \label{lem:decomp-1}
Suppose $\gamma = 1$ and  $\ell \geq 1$ is a constant. An ($\gamma, \ell$)-decomposition with $L = O(\log n)$ of a tree $G$ can be computed in $O(\log n)$ rounds deterministically.
\end{lemma}

In this paper, we are only interested in the case of $\gamma \gg 1$.

\begin{lemma}\label{lem-decomp-2}
Suppose \(k=O(1)\), \(\ell=O(1)\), and
\[
    \gamma \geq n^{1/k}(2\ell)^{1-1/k}.
\]
A \((\gamma,\ell)\)-decomposition with \(L\leq k\) of a tree \(G\) can be computed deterministically in $O(\gamma+\log^\ast n)$ rounds. In particular, if \(\gamma=\Theta(n^{1/k})\), then the round complexity is \(O(n^{1/k})\).
\end{lemma}

\begin{proof}
Fix an arbitrary vertex \(v\in V\). For each iteration \(i\), let \(S_i\) be the connected component containing \(v\) in the subgraph induced by the current set \(U\) at the beginning of the \(i\)th iteration. If \(v\notin U\), then we set \(S_i=\emptyset\). Similarly, let \(S_i'\) be the connected component containing \(v\) at the beginning of the \compress\ operation in the \(i\)th iteration, again setting \(S_i'=\emptyset\) if \(v\notin U\). Observe that \(|S_1| = |V| = n\).

We show that \(S_k'=\emptyset\). Consider an iteration \(i\) with \(S_i'\neq \emptyset\). Let \(A\) be the number of degree-\(2\) vertices in \(S_i'\) that are not removed during the \(i\)th \compress\ operation, and let \(B\) be the number of vertices in \(S_i'\) whose degree in \(S_i'\) is not \(2\). Then
\[
    |S_{i+1}| \leq A+B.
\]

The degree-\(2\) vertices in \(S_i'\) that are not removed during the \(i\)th \compress\ operation form maximal paths of length at most \(\ell-1\). Contracting each such maximal path gives a tree whose vertex set consists of the \(B\) vertices whose degree in \(S_i'\) is not \(2\). Hence there are at most \(B-1\) such maximal paths, and so
\[
    A \leq (\ell-1)(B-1).
\]
Therefore,
\[
    |S_{i+1}| \leq A+B  < \ell B.
\]

It remains to upper bound \(B\). Let \(Q\) be the number of vertices in \(S_i'\) whose degree in \(S_i'\) is at most \(1\). Since \(S_i'\) is a tree, we have
\[
    Q \geq \frac{B}{2}.
\]
Indeed, if \(S_i'\) consists of a single vertex, then this is immediate. Otherwise, the number of leaves in a tree is at least the number of vertices of degree at least \(3\), and hence at least half of the vertices whose degree is not \(2\).

Now consider any vertex \(x\in S_i'\) whose degree in \(S_i'\) is at most \(1\). Since \(x\) survives all \(\gamma\) \rake\ operations in iteration \(i\), the vertices of \(S_i\setminus S_i'\) in the components attached to \(x\) must contain at least \(\gamma\) vertices. Otherwise, all these attached vertices would have been removed within fewer than \(\gamma\) \rake\ operations, after which \(x\) would have degree at most \(1\) and would also be removed before the \compress\ operation. For distinct such vertices \(x\), these attached sets are disjoint, since \(S_i\) is a tree. Hence
\[
    Q\gamma \leq |S_i|.
\]
Together with \(B\leq 2Q\), this gives
\[
    B \leq \frac{2|S_i|}{\gamma}.
\]
Consequently,
\[
    |S_{i+1}| < \ell B \leq \frac{2\ell}{\gamma}|S_i|.
\]

Applying this inequality for \(i=1,2,\ldots,k-1\), we obtain
\[
    |S_k|
    \leq
    \left(\frac{2\ell}{\gamma}\right)^{k-1} n.
\]
By the assumption \(\gamma \geq n^{1/k}(2\ell)^{1-1/k}\), we have
\[
    \left(\frac{2\ell}{\gamma}\right)^{k-1} n
    \leq \gamma.
\]
Thus \(|S_k|\leq \gamma\). Therefore all vertices of \(S_k\) are removed during the \(\gamma\) \rake\ operations in iteration \(k\), and so \(S_k'=\emptyset\).

Since \(v\) was arbitrary, no vertex remains at the beginning of the \compress\ operation in iteration \(k\). Hence no iteration after \(k\) is needed, and the resulting decomposition has \(L\leq k\).

The main part of the algorithm takes \(O(k(\gamma+\ell)) = O(\gamma)\) rounds. The post-processing step takes \(O(\log^\ast n)\) rounds because \(\ell=O(1)\). This proves the claimed round complexity.
\end{proof}

\paragraph{The set of paths $\Pset$.}
We revisit the proof in Section~\ref{sect-gap-review} and prove that the required set of paths $\Pset$ can be computed in $O(\sqrt{n})$ rounds.
We run our algorithm for constructing a ($\gamma, \ell$)-decomposition 
with the parameters $\gamma = \tau = n^{1/2} (2\ell)^{1/2} = \Theta(\sqrt{n})$ and $\ell = \Lpump = \Theta(1)$. Here $\tau = \Theta(\sqrt{n})$ is the parameter in the definition of the skeleton tree $G_{\sk}$ in Section~\ref{sect-gap-review}. Remember that $G_{\sk}$ is the result of iteratively removing all leaf vertices of $G$ for $\tau$ iterations.

By Lemma~\ref{lem-decomp-2}, our ($\gamma, \ell$)-decomposition   satisfies $L = 2$, and it decomposes $V$ into three sets $\VR{1}$, $\VC{1}$, and $\VR{2}$, and the decomposition can be computed in $O(\sqrt{n})$ rounds.
It is clear from the description of the algorithm that $G_{\sk} = \GC{1}$ is exactly the subgraph induced by $\VC{1} \cup \VR{2}$. Selecting $\Pset$ as the set of all connected components of $\VC{1}$ satisfies all the requirements of  $\Pset$ stated in Section~\ref{sect-gap-review}.

%% file: LCLgapExtend.tex
\section{Extension to Other Gaps \label{sect:gap-extend}}
In this section, we  prove   Theorem~\ref{thm-gap-main}  by extending the  proof of the $\omega(n^{1/2})$--$o(n)$ gap by Balliu et al.~\cite{Balliu2018disc} reviewed in Section~\ref{sect-gap-review}. Here we only focus on the case of deterministic algorithms. The complete proof of Theorem~\ref{thm-gap-main} is left to Appendix~\ref{sect-gap-details}.


\subsection{Proof Idea}\label{sect-high-level}
Let $k \geq 2$ be any positive constant,
For any
given  $o(n^{1/(k-1)})$-round deterministic  algorithm $\mathcal{A}$ for a  given LCL problem $\mathcal{P}$, our goal is to design a new deterministic  algorithm $\mathcal{A}'$ with round complexity $O(n^{1/k})$.

We compute a ($\gamma, \ell$)-decomposition, with $\gamma = \Theta(n^{1/k})$ and $\ell \geq \Lpump$, to decompose $V$ into the subsets \[\VR{1}, \VC{1}, \VR{2}, \ldots, \VR{k-1}, \VC{k-1}, \VR{k},\]  and  then apply the pumping lemma to extend each path in
$\VC{1}, \VC{2} \ldots, \VC{k-1}$ to a path of length within $[w, w+\Lpump]$,
in order to produce a virtual tree  with $O(w^{k-1})$ vertices, omitting the dependence on $n$. 
By Lemma~\ref{lem-decomp-2},  such a decomposition can be computed in $O(n^{1/k})$ rounds.

If we select $w$ to be sufficiently large, the execution of a given  $o(n^{1/(k-1)})$-round algorithm $\mathcal{A}$ takes less than $0.1w$ rounds on the virtual tree. As each connected component induced by $\VR{i}$ is a rooted tree with diameter $O(\gamma) = O(n^{1/k})$, the simulation of $\mathcal{A}$ can be done in $O(n^{1/k})$ rounds in the underlying network $G$. Hence we obtain
an $O(n^{1/k})$-round algorithm $\mathcal{A}'$ for the same problem.

The above high-level approach does not  work immediately, as we will encounter some issues described below, but these issues can be overcome using the graph operations defined in~\cite{ChangP17}.

\paragraph{An issue in pumping bipolar subtrees.} 
The reason that we can apply the pumping lemma for $\VC{1}$ in Section~\ref{sect-gap-review} is that  each connected component $P$ of $\VC{1}$ is naturally associated with a bipolar tree $P^\ast$.
We do not have this property for the connected components of $\VC{i}$ for $i > 1$, since each vertex $v \in   \VR{1} \cup \VC{1} \cup \cdots \cup \VR{i} =  V(G) \setminus V(\GC{i})$ might be reachable to more than one connected component of $\VC{i}$ via the vertices in $V(G) \setminus V(\GC{i})$. See Figure~\ref{fig:treedecomp}.

Let us recall the virtual tree construction in  Section~\ref{sect-gap-review}.
Define $G'$ as the graph resulting from pumping the paths in $\VC{1}$. Formally,  for each connected component $P$ of  $\VC{1}$,  replace $P^\ast$ by a much longer bipolar tree $P'$ with $\type(P^\ast) = \type(P')$, where $P^\ast$ is the bipolar subtree of $G$ induced by the vertices in $P$ and all vertices in  $\VR{1}$ that are reachable to a vertex in $P$ via the vertices in  $\VR{1}$.
Note that $G'$ is the same as the virtual tree $G_\vi$  in Section~\ref{sect-gap-review}, if we use $G_{\sk} = \GC{1}$  and let $\Pset$  be the set of connected components of $\VC{1}$.

For the $k > 2$ case,  we would like to also pump the paths in $\VC{2}$ in this graph $G'$ in a way similar to the case of $\VC{1}$.
As discussed above, a difference between $\VC{1}$ and $\VC{2}$ is that it is possible that a vertex $v \in V(G') \setminus V(\GC{2})$ is reachable to multiple connected components in $\VC{2}$ via the vertices in $V(G') \setminus V(\GC{2})$, so we are unable to associate a bipolar tree $P^\ast$ to each connected component $P$ of $\VC{2}$.


Recall that in our high-level proof idea we will ultimately simulate an algorithm $\mathcal{A}$ on a virtual tree, and the runtime of $\mathcal{A}$ will be less than $0.1w$. Again consider the graph $G'$ and one of its bipolar subtree $P'$ resulting from pumping $P^\ast$ for some    connected component $P$ of $\VC{1}$. The virtual bipolar tree  $P'$ separates the graph $G'$ into two parts, and the vertices in one part does not need to communicate with the vertices in the other part in the  simulation of $\mathcal{A}$, as its runtime is less than $0.1w$. Recall that the core path of   $P'$ has at least $w$ vertices.

Motivated by the above discussion, we consider the  graph $G''$ defined as the result of applying the following operation on $G'$ for each virtual bipolar tree $P'$. Let $u$ and $v$ be the two vertices in $V(G') \setminus V(P')$ adjacent to the two poles $s$ and $t$ of $P'$ via the edges $\{u,s\}$ and $\{v,t\}$. We \emph{duplicate} $P'$ into two identical bipolar subtrees, one is attached to $u$ via $\{u,s\}$, the other is attached to $v$ via $\{v,t\}$. Note that this is the $\cut$ operation defined in~\cite{ChangP17}. See Figure~\ref{fig:split} for an illustration.

\begin{figure}[ht!]
\centering
{\includegraphics[width=1\textwidth]{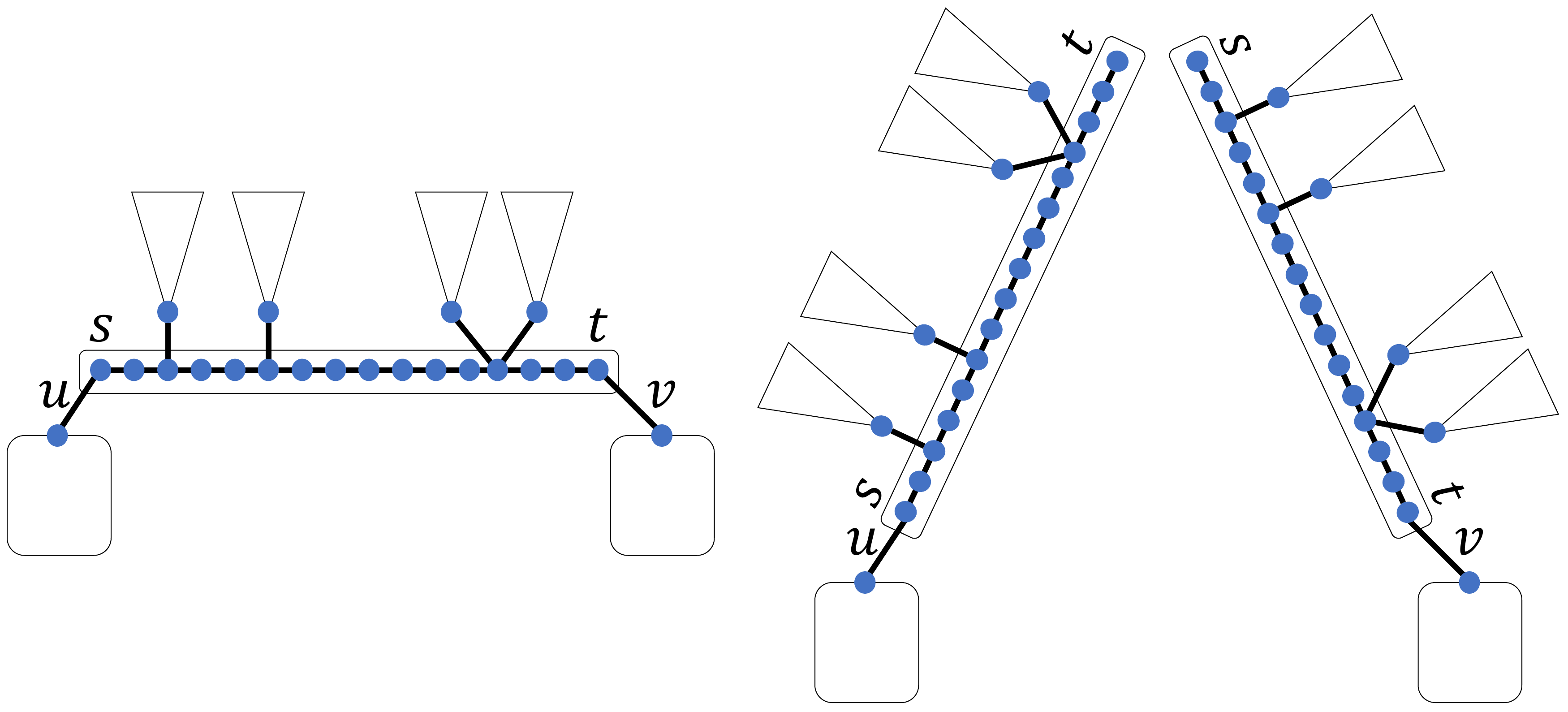}}
\caption{\label{fig:split} The $\cut$ operation}
\end{figure}

Let $P$ be a connected component of $\VC{2}$ in the graph $G''$. With respect to $G''$, we are able to define $P^\ast$ in the same way as the case of $\VC{1}$.
Specifically, we define $P^\ast$
as the bipolar subtree of $G''$ induced by the vertices in $P$ and all vertices in   $V(G'') \setminus V(\GC{2})$ that are reachable to a vertex in $P$ via the vertices in   $V(G'') \setminus V(\GC{2})$. 
Using this approach recursively, we can pump the paths in all layers $\VC{1}, \VC{2}, \ldots, \VC{k-1}$.

\paragraph{An issue caused by duplicating bipolar trees.} 
The duplication of bipolar subtrees in $\cut$ also causes an issue. Consider the graphs $G$, $G'$, and $G''$ defined above. As discussed in Section~\ref{sect-gap-review},  given a legal labeling of $G'$, we can obtain a legal labeling of $G$ using  a property  of $\simm$ and the fact that pumping does not alter the type of a bipolar tree. 
However, when we try to obtain a legal labeling $\LL'$ of $G'$ from a  given  legal labeling $\LL''$ of $G''$, we encounter an issue that the two copies of a bipolar subtree $P'$ resulting from applying $\cut$ in $G'$ might be labeled differently in  $\LL''$.

To resolve this issue, before the duplication of $P'$ in the construction of $G''$ from $G'$, we  let some vertices near the middle of $P'$ to first commit to a certain  labeling. Such a labeling is computed by simulating the given $o(n^{1/(k-1)})$-round algorithm $\mathcal{A}$, pretending that the number of vertices is $O(w^{k-1})$, omitting the dependence on $n$. We can  assume that the runtime of $\mathcal{A}$ on $P'$ is at most $0.1w$ by selecting $w$ to be sufficiently large.

Specifically, let $P^\ast$ be a bipolar tree that we would like to apply the pumping lemma. 
We give a different way of constructing $P'$ from $P^\ast$.
We write $P^\ast$ as a string of subtrees
 $(\TT_1, \TT_2,  \ldots, \TT_x)$. 
Let $(v_1, v_2, \ldots, v_x)$ be the core path of $P^\ast$
and $e = \{v_{\lfloor x/2 \rfloor}, v_{\lfloor x/2 \rfloor+1}\}$ be the middle edge of the core path.
Consider the decomposition $P^\ast = \XX \circ \YY \circ \ZZ$, where
$\YY = (\TT_{\lfloor x/2 \rfloor-r+1}, \ldots, \TT_{\lfloor x/2 \rfloor + r})$ is the middle part.
We apply the pumping lemma on $\XX$ and $\ZZ$ to extend them to longer bipolar trees whose size of core path is within $[w, w+\Lpump]$, and then we assign output labels to the vertices in $N^{r-1}(e) = N^{r-1}(v_{\lfloor x/2 \rfloor}) \cup N^{r-1}(v_{\lfloor x/2 \rfloor+1})$ by simulating the algorithm $\mathcal{A}$, whose runtime is less than $0.1w$. The resulting partially labeled bipolar tree is $P'$. Note that this  construction of $P'$ from $P^\ast$ is the same as the one in~\cite{ChangP17} using the operations $\labelling$ and $\extend$. In this paper, we call this operation $\ext$. See Figure~\ref{fig:labelextend} for an illustration.

\begin{figure}[ht!]
\centering
{\includegraphics[width=1\textwidth]{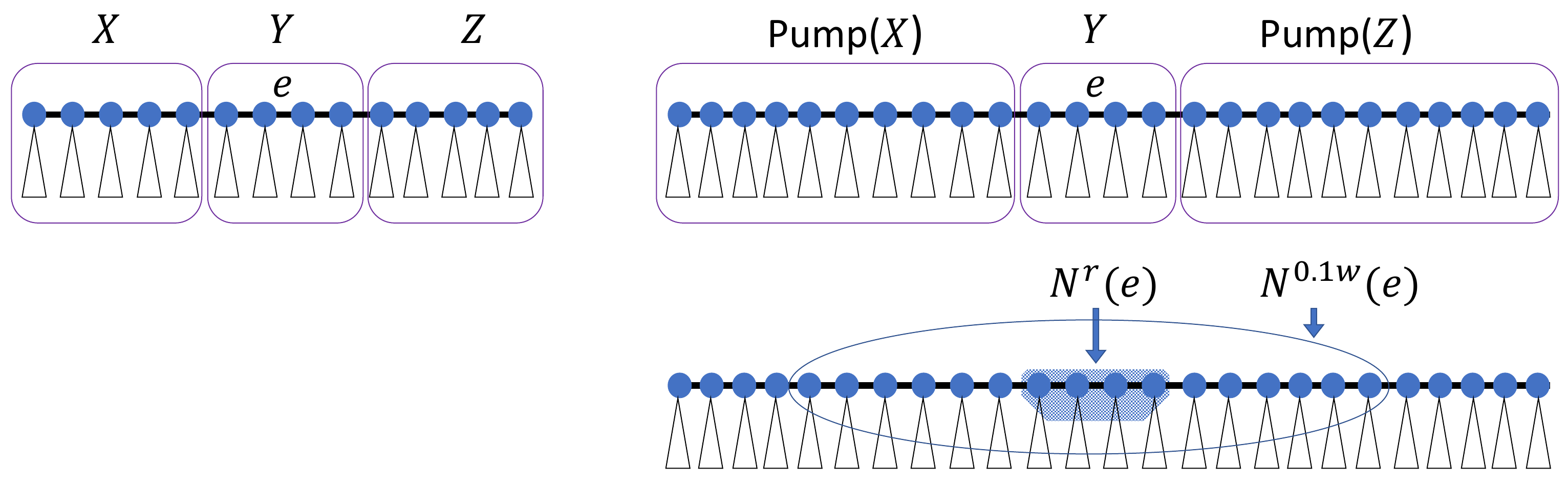}}
\caption{\label{fig:labelextend} The $\ext$ operation}
\end{figure}

We briefly explain why doing this labeling of middle vertices resolves the issue. Suppose $G_a$ is the result of apply $\cut$ to some bipolar subtree $P'$ in $G_b$, where this bipolar tree $P'$ is constructed as above and its middle vertices have been assigned output labels. In a given legal labeling of $G_a$, the two copies of $P'$ might be labeled differently, but their middle vertices must be labeled the same. We decompose $P'= P_s \circ P_t$ into two parts by cutting along the middle edge $e$. The pole $s$ is in $P_s$, and the other pole $t$ is in $P_t$. We name the two copies of $P'$ in $G_a$ by $P_s'$ and  $P_t'$ based on the poles they use to connect to the rest of the graph.
To obtain a legal labeling of $G_b$ from a given legal labeling of $G_a$,
we  simply label $P_s$ by adapting the labeling of  $P_s'$ in $G_a$, and label $P_t$ by adapting the labeling of  $P_t'$ in $G_a$. The legality of the resulting labeling of $G_b$ is easy to verify.

See Figure~\ref{fig:recoverlabel} for an example. The top-left figure illustrates the bipolar subtree $P'$ in $G_b$, where its middle vertices have been assigned output labels. The top-right figure illustrates the graph $G_a$, which results from applying $\cut$ to $P'$ in $G_b$. The down-right figure illustrates the given legal labeling of $G_a$. The down-left figure shows the   legal labeling of $G_b$ obtaining from the given legal labeling of $G_a$.

\begin{figure}[ht!]
\centering
{\includegraphics[width=1\textwidth]{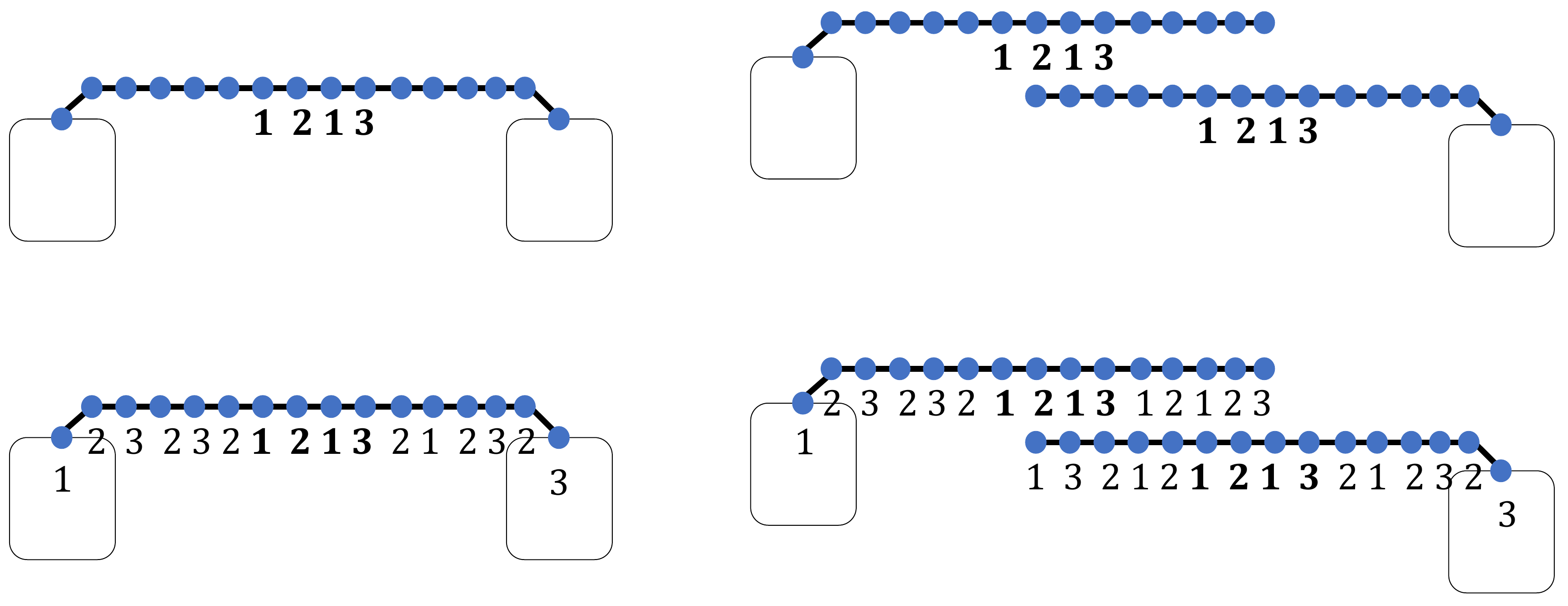}}
\caption{\label{fig:recoverlabel} Obtaining a legal labeling}
\end{figure}

\subsection{A Sequence of Virtual Graphs}

The approach discussed above naturally leads to a sequence of partially labeled virtual graphs \[\RRR{1}, \RRC{1}, \RRR{2}, \RRC{2}, \ldots, \RRR{k}.\] Each virtual graph has a \emph{real} and an \emph{imaginary} part.  Each real vertex corresponds to a vertex in the underlying network $G$.
The graph $\RRR{i}$ will have $\GR{i}$ as its real part, and the graph $\RRC{i}$ will have $\GC{i}$ as its real part. 
The imaginary part of these graphs are subtrees attached to the real vertices. 
In the actual distributed implementation, the simulation of the imaginary subtrees attached to a real vertex $v$ are handled by $v$ in the underlying network $G$.

\paragraph{Construction of $\RRR{1}$.} The graph $\RRR{1}$ is unlabeled and it equals the underlying network $G$.

\paragraph{Construction of $\RRC{1}$.} The graph  $\RRC{1}$ is almost identical to  $G = \RRR{1}$. In $\RRC{1}$, only the vertices in  $\GC{1}$ are real.
For each connected component $S$ of $\VR{1} = \RRR{1} \setminus \GC{1}$, there is at most one vertex  $v \in \GC{1}$ that is adjacent to $S$.  If such a vertex $v$ exists, then $S$ becomes an imaginary subtree stored in the real vertex $v$. Otherwise,  $S$ is not included in $\RRC{1}$.

\paragraph{Construction of $\RRR{2}$.} The graph $\RRR{2}$ is the graph $G'$ in   Section~\ref{sect-high-level}. Formally, the graph  $\RRR{2}$ is constructed by applying the following operation to each connected component $P = (v_1, v_2, \ldots, v_x)$ of $\VC{1}$ in  $\RRC{1}$. Let $P^\ast = (\TT_1, \TT_2,  \ldots, \TT_x)$ be the bipolar subtree induced by $P$ and the imaginary subtrees therein.  
Replace $P^\ast$ by the partially labeled bipolar tree $P'$ which is the result of applying the   operation $\ext$ to
$P^\ast$, and then apply $\cut$ to $P^\ast$. 

\paragraph{Construction of $\RRC{2}$.}  
 The graph  $\RRC{2}$ is almost identical to  $\RRR{2}$. In $\RRC{2}$, only the vertices in  $\GC{2}$ are real.
Due to the $\cut$ operation, for each connected component $S$ of $\RRR{2} \setminus \GC{2}$, there is at most one vertex  $v \in \GC{2}$ that is adjacent to $S$.  If such a vertex $v$ exists, then $S$ becomes an imaginary subtree stored in the real vertex $v$. Otherwise,  $S$ is not included in $\RRC{2}$.

\paragraph{Construction of the other graphs.}
The rest of the partially labeled graphs $\RRR{3}$, $\RRC{3}$, $\RRR{4}$, $\RRC{4}$, $\ldots$, $\RRR{k}$ are constructed analogously.
In the end, $\RRR{k}$ is a virtual graph with $O(w^{k-1})$ vertices, omitting the dependence on $n$. It is straightforward to see that the sequence of virtual graphs $\RRR{1}, \RRC{1}, \RRR{2}, \RRC{2}, \ldots, \RRR{k}$  can be constructed in $O(n^{1/k})$ rounds.

\paragraph{Completing the labeling.}
Recall that the partial labelings of   $\RRR{1}, \RRC{1}, \RRR{2}, \RRC{2}, \ldots, \RRR{k}$ are computed using the operation $\ext$, which is based on simulating $\mathcal{A}$ while assuming that the number of vertices is $n' = O(w^{k-1})$, omitting the dependence on $n$. By the correctness of $\mathcal{A}$, each of these partial labelings can be completed into a complete legal labeling.


Since each connected component of the real part of $\RRR{k}$ has at most $O(n^{1/k})$ vertices, a complete legal labeling of $\RRR{k}$ can be found in $O(n^{1/k})$ rounds by a brute-force information gathering. 
Once we have a complete labeling of $\RRR{k}$, we can start from this complete labeling to obtain a complete legal labeling for  $\RRC{k-1},   \RRR{k-1}, \RRC{k-2}, \ldots, \RRR{1}=G$ in $O(n^{1/k})$ rounds  in view of the discussion in Section~\ref{sect-high-level}, as these graphs are constructed by applying $\ext$ and then applying $\cut$ to the bipolar trees resulting from $\ext$. 

The round complexity for finding a legal labeling of $G = \RRR{1}$ using this approach is $O(n^{1/k})$ because the size of each connected component of $\VR{i}$ is $O(n^{1/k})$.

Hence we have the $\omega(n^{1/k})$--$o(n^{1/(k-1)})$ gap for LCL problems on bounded-degree trees for the case of deterministic algorithms. That is,  given a deterministic  $o(n^{1/(k-1)})$-round algorithm $\mathcal{A}$ for $\mathcal{P}$, we can construct another deterministic algorithm $\mathcal{A}'$ that takes $O(n^{1/k})$ rounds.

In 
the appendix we provide a more formal proof for this result that allows us to handle the case of randomized algorithms and offers a  sequential algorithm that given an integer $k \geq 2$ and a description of an LCL problem $\mathcal{P}$, decides whether its (randomized and deterministic) round complexity is $O(n^{1/k})$ or  $\Omega(n^{1/(k-1)})$.

\paragraph{A note about unique identifiers.}  A subtle issue about the simulation of $\mathcal{A}$ is that the simulation needs distinct identifiers. Specifically, to guarantee the correctness of a deterministic $\tau$-round algorithm for an LCL problem with locality radius $r$, it suffices that any two vertices within distance $2\tau + 2r$ have distinct identifiers~\cite{ChangKP16}. 

We only simulate $\mathcal{A}$ when we apply $\ext$.
When we do the simulation of  $\mathcal{A}$, we can locally generate distinct identifiers of length $O(\log n')$ for all vertices in $N^{0.1w+r}(e)$, where $e$ is the middle edge of the core path of the bipolar tree on which we run  $\mathcal{A}$, and $n' = O(w^{k-1})$, omitting the dependence on $n$. 
This partial ID assignment satisfies the requirement that, for any two vertices $u$ and $v$ that are assigned identifiers and are within distance $2 \cdot 0.1w + 2r$, their identifiers are distinct.




%% file: hardness.tex
\section{EXPTIME Hardness}\label{sect:hardness}

The goal of this section is to prove Theorem~\ref{thm-hardness}. We do the reduction from the following well-known $\EXPTIME$-complete problem~\cite{arora2009computational,ChandraKS81}: given a linear space-bounded alternating Turing machine $\mathcal{M}$ and a binary string $x$,  decide whether $\mathcal{M}$ accepts $x$. 

Fix two LCL problems $\mathcal{P}_1$ and $\mathcal{P}_2$ on bounded-degree trees. Assume that the round complexity of $\mathcal{P}_1$ is $T_1(n)$,  the round complexity of $\mathcal{P}_2$ is $T_2(n)$, and $T_1(n) \ll T_2(n)$. We will construct an LCL problem $\mathcal{P}_{\MM,x}$ such that its round complexity is $\Theta(T_1(n))$ if $\mathcal{M}$ accepts $x$ and  is $\Theta(T_2(n))$ otherwise. This construction is a natural extension of the  construction in the $\PSPACE$-hardness proof  in~\cite{Balliu2019decidable}.

\paragraph{Alternating Turing machines.}  An \emph{alternating Turing machine}  $\mathcal{M}$ is defined by the following parameters.
\begin{itemize}
    \item $Q$ is the set of states, and there are four possible types of a state: $\exists$, $\forall$, $\acc$, and $\rej$.
    \item $\Gamma=\{0,1\}$ is the tape alphabet.
    \item $\delta_1$ and $\delta_2$ are two transition functions mapping $Q \times \Gamma$ to  $Q \times \Gamma \times \{L,R\}$. 
    \item $q_0 \in Q$ is the initial state.
\end{itemize}

In the above definition, $\delta_i(q_1, \alpha_1) = (q_2, \alpha_2, X)$ indicates the following transition. If the machine is currently in state $q_1$ and the  head reads the symbol $\alpha_1$, then the machine does the following: (i) write $\alpha_2$ to the current cell of the tape, (ii) move  the head  left  or  right  according to $X \in\{L, R\}$, and (iii) transition to state $q_2$.

\paragraph{Configuration.} As we only consider linear space-bounded alternating Turing machines with the binary alphabet, a \emph{configuration} consists of the following.
\begin{itemize}
    \item A length-$s$ binary string $S$ representing the current tape, where $s$ is the space upper bound.  
    \item An index $i \in [s]$ indicating the  current position of the head.
    \item A state $q \in Q$ indicating the current state.
\end{itemize}

In the \emph{initial configuration} $C_0=(S,i,q)$, $S = x$ is the input string $x$, $i = 1$ is the initial position of the head, and $q = q_0 \in Q$ is the initial state.

A configuration where $q$  is an $\acc$-state is said to be \emph{accepting}.
A configuration where $q$  is a $\rej$-state is said to be \emph{rejecting}. An execution of a Turing machine terminates once it enters an $\acc$-state or a $\rej$-state.

For the case $q$ is an $\forall$-state, the configuration $C$ is accepting if the two configurations reachable from $C$ in one step  are both accepting, and $C$  is rejecting if at least one configuration reachable from $C$ in one step is rejecting.
For the case $q$ is an $\exists$-state, the configuration $C$ is accepting if at least one configuration reachable from $C$ in one step  is accepting, and $C$  is rejecting if the two configurations reachable from $C$ in one step  are both  rejecting.

To avoid infinite loop, we can modify the given $s$-space bounded alternating Turing machine $\MM$ by implementing a counter that forces $\mathcal{M}$ to enter a $\rej$-state after $s \cdot 2^{s + |Q|}$ steps, because $s \cdot 2^{s + |Q|}$ is an upper bound on the number of configurations for an $s$-space bounded alternating Turing machine. 

From now on we assume that each configuration is either accepting or rejecting.
We say that $\mathcal{M}$ accepts $x$ if the initial configuration is accepting.

\paragraph{The tree representation.} Given an $s$-space bounded alternating Turing machine $\mathcal{M}$ and a binary string $x \in \{0,1\}^s$, the execution of  $\mathcal{M}$ on $x$ can be encoded as a rooted tree $T_{\MM,x}$ with input labeling, as follows.

Let $C = (S,i,q)$ be a configuration, and let $z \in \{1,2\}$, define $P_{C,z}=(v_1, v_2, \ldots, v_s)$ as the path where $v_j$ is labeled by the 5-tuple \[\left(a = S[j],  \ q, \  b = \mathbf{1}(i=j), \ y\in\{\acc,\rej\}, \ z\right).\]
Here $b = \mathbf{1}(i=j)$ is the indicator function such that $b = 1$ if $i = j$, and $b = 0$  otherwise. The parameter $y$ indicates whether $C$ is accepting or rejecting.  The parameter $a = S[j] \in \{0,1\}$ represents the $j$th element in the tape $S$.

We define the rooted tree $T_{C,z}$ recursively as follows.
\begin{itemize}
    \item If $C$ is a configuration where $q$ is an $\acc$-state or a $\rej$-state, then $T_{C,z}$ is simply the path $P_{C,z}$, where the first vertex is the root.
    \item Suppose $C$ is a configuration where $q$ is a $\forall$-state or a $\exists$-state.  Let $C_1$ and $C_2$ be the configurations reachable from $C$ in one step via $\delta_1$ and $\delta_2$, respectively.  The rooted tree $T_{C,z}$ is the result of attaching $T_{C_1,1}$ and $T_{C_2,2}$ to the last vertex of $P_{C,z}$, and we set the first vertex of $P_{C,z}$ as the root.
\end{itemize}
Then $T_{\MM,x}$ is defined as $T_{C_0,1}$, where $C_0$ is the initial configuration. Note that $T_{\MM,x}$  is uniquely defined given $\MM$ and $x$.

\subsection{LCL Construction}
The  idea of the construction of $\mathcal{P}_{\MM,x}$ is as follows.
 We will use input labels to ensure that the input graph $G=(V,E)$ is of the following form. The vertex set $V$ consists of a \emph{main part} $V_M$ and an \emph{auxiliary part} $V_A$.
 Whether $v \in V$ belongs to $V_M$ or $V_A$ depends on its input label.
For the sake of simplicity,  we assume that each edge is initially oriented. We are allowed to make this assumption because edge orientation can also be encoded as input labels on vertices, as discussed in Section~\ref{sect-prelim}. The orientation of an edge $\{u,v\}$ from $u$ to $v$ is denoted as $u \rightarrow v$.

\paragraph{Basic requirements.}
We make the following requirements.
\begin{itemize}
    \item Each  $v \in V_A$ has at most one neighbor $u \in N(v)$ such that $v \rightarrow u$. 
    \item Each  $v \in V_M$ has at most one neighbor $u \in N(v)$ such that $u \in V_A$.
    \item For each edge $\{u,v\}$ with $u \in V_A$ and $v \in V_M$, it is oriented in the direction $u \rightarrow v$.
\end{itemize}
That is,  each vertex $v \in V_M$ in the main part is associated with \emph{at most one} \emph{auxiliary subtree} $T_v$ whose root $u$ is oriented to $v$, and all vertices in $T_v$ are in the auxiliary part $V_A$.

 The input graph might be an arbitrary tree of maximum degree $\Delta$ that is not necessarily in this form. However, whether $v \in V$ is in the correct form can be checked locally by examining its neighborhood $N(v)$, as the above basic requirements are {locally checkable}. Using radius-1 constraints, we can force the vertices that do not meet these requirements to mark themselves a special symbol $\sun$ in their output labeling, indicating that they should be treated as if they do not exist.
 From now on we assume that all vertices in the graph $G$ meet these requirements.

 
\paragraph{Good and bad vertices.}
For each vertex $v \in V_M$ in the main part, it belongs to one of the following two cases. Note that if $\MM$ accepts $x$, then all vertices $v \in V_M$ must be good. 
\begin{itemize}
    \item We say that $v \in V_M$ is \emph{good} if at least one of the following conditions is met.
    \begin{itemize}
        \item There is no auxiliary tree $T_v$ attached to $v$. 
        \item $T_v = T_{\MM, x}$ and   $\MM$ accepts $x$.
        \item $T_v \neq T_{\MM, x}$.
    \end{itemize}
    \item We say that $v \in V_M$ is \emph{bad} if $T_v = T_{\MM, x}$ and   $\MM$ rejects $x$.
\end{itemize}

Our goal is to design $\mathcal{P}_{\MM,x}$ in such a way that each good vertex $v \in V_M$ is able to mark itself a special symbol $\bigstar$ in $O(1)$ rounds of communication, and each bad vertex $v \in V_M$ is not able to
do so in the sense that if a bad vertex  $v \in V_M$  marks itself $\bigstar$, then the overall labeling must be illegal.

All vertices in $V_M$ that are marked $\bigstar$ are required to solve the easier problem  $\mathcal{P}_1$, which has round complexity $T_1(n)$. All remaining vertices in $V_M$ are required to solve the more difficult problem $\mathcal{P}_2$, which has round complexity $T_2(n)$.

If $\mathcal{M}$ accepts $x$, then   the round complexity of  $\mathcal{P}_{\MM,x}$ is $\Theta(T_1(n))$, because all vertices $v \in V_M$ are good.
If $\mathcal{M}$ rejects $x$,  then 
the round complexity of  $\mathcal{P}_{\MM,x}$ is $\Theta(T_2(n))$, because it is possible that all vertices $v \in V_M$ are bad. Therefore, the requirements of  Theorem~\ref{thm-hardness} are met.

\subsection{Detecting Errors in Tree Representations}

For each vertex $v \in V_M$, if there is no $T_v$ attached to $v$, or if the input label at the root of $T_v$ indicates $y = \acc$, then $v$ is allowed to mark itself  $\bigstar$. This can be done using a radius-1 constraint in $\mathcal{P}_{\MM,x}$.   Note that if the input label at the root of $T_v$ indicates $y = \acc$, then  either (i) $T_v = T_{\MM, x}$ and   $\MM$ accepts $x$, or (ii)  $T_v \neq T_{\MM, x}$. In both cases, $v$ is good. 

The remaining situation is when the input label at the root of $T_v$ indicates $y \neq \acc$. 
To deal with this situation, we need to design locally checkable constraints that allow $v$ to mark itself $\bigstar$ if and only if $T_v \neq  T_{\MM, x}$. Furthermore, $v$ need to be able to do this in $O(1)$ rounds deterministically if  $T_v \neq  T_{\MM, x}$.

\paragraph{The error symbol.}
To do so, we introduce a  symbol $\err$ with the following property. Suppose $v \in V_A$ is marked $\err$ and $v \rightarrow u$. If $u \in V_A$, then $u$ can also mark itself $\err$. If $u \in V_M$, then $u$ can mark itself $\bigstar$. This property can be specified as  radius-1 constraints in $\mathcal{P}_{\MM,x}$, and it allows us to mark $u \in V_M$ as $\bigstar$ when there is a vertex $v \in T_u$ marking itself $\err$.

For the rest of the proof, we focus on designing locally checkable rules that enable at least one vertex $u \in T_v$ 
to mark itself $\err$ if and only if $T_v \neq  T_{\MM, x}$. 
We will go though each possible error in the tree representation that we can have, and show how such an error can be captured using locally checkable rules.

\paragraph{Basic errors.}  We begin with the following basic errors.
\begin{itemize}
    \item Consider the case that $v \in V_A$ has no children. If its input label indicates that $q$ is not an $\acc$-state or a $\rej$-state, then $v$ can mark itself $\err$. Since $v$ is a leaf of the tree representation, if the representation is correct, then $v$ must be in a configuration where the state is an $\acc$-state or a $\rej$-state.
    \item Consider the case that $v \in V_A$ has exactly one child $u \in V_A$.  We require that the parameters $q$, $y$, and $z$ are the same in both $u$ and $v$, otherwise $v$ can mark itself $\err$. The reason is that if the representation is correct, $u$ and $v$ must belong to the same path $P_{C,z}$ representing a configuration $C$ and a number $z \in \{1,2\}$.
    \item Consider the case that $v \in V_A$ has more than two children. Then $v$ can mark itself $\err$, because in a correct tree representation all vertices have at most two children.
\end{itemize}

\paragraph{Errors in path lengths.}
To ensure that each path $P_{C,z}$ contains exactly $s$ vertices, we introduce $s$ auxiliary symbols $\ell_1, \ell_2, \ldots, \ell_s$. Any vertex $v \in V_A$ with zero or two children can mark itself $\ell_1$. For $i = 1, 2, \ldots, s-1$, any vertex $v \in V_A$ with a child marked $\ell_{i}$ can mark itself $\ell_{i+1}$. The possible errors are as follows.
\begin{itemize}
    \item Let $v \in V_A$ be a vertex marked $\ell_i$ for some $2 \leq i \leq s$. We require $v$ to have exactly one child.
    \item Let $v \in V_A$ be a vertex that has a child $u$ marked $\ell_s$. We require $v$ to have exactly two children.
\end{itemize}
  If a requirement is not met, then  $v$  can mark itself $\err$.

\paragraph{Errors in $y$-labels.}
Recall that the parameter $y \in \{\acc,\rej\}$ in an input label indicating whether the corresponding configuration is accepting or rejecting. The possible errors about $y$-labels are as follows.
\begin{itemize}
    \item Let $v \in V_A$ be a vertex. We require that its $y$-label and $q$-label are consistent in the following sense. If $q$ is an $\acc$-state, then $y = \acc$. If $q$ is a $\rej$-state, then $y = \rej$.  
    \item Let $v \in V_A$ be a vertex of two children.  We require that its $y$-label and $q$-label are consistent with the $y$-labels in the two children in the following sense. Let $y_1$ and $y_2$ be the $y$-labels of the two children. For the case $q$ is a $\forall$-state, we have $y = \acc$ if $y_1 = \acc$ and $y_2 = \acc$, otherwise $y = \rej$.    For the case $q$ is a $\exists$-state, we have $y = \rej$ if  $y_1 = \rej$ and $y_2 = \rej$, otherwise $y = \acc$.   
\end{itemize}
  If a requirement is not met, then  $v$  can mark itself $\err$.

\paragraph{Errors in $z$-labels.}
Recall that the parameter $z \in \{1, 2\}$ in an input label indicating whether the corresponding configuration is the result of applying $\delta_1$ or $\delta_2$ to the parent configuration.  For the special case of the initial configuration, we have $z = 1$.
The possible errors about $z$-labels are as follows.

\begin{itemize}
    \item Let $v \in V_A$ be a root vertex, i.e., $v \rightarrow u$ for some $u \in V_M$. We require that its $z$-label is $1$.   
    \item Let $v \in V_A$ be a vertex of two children.  We require that the two $z$-labels $z_1$ and $z_2$ in its two children are distinct.  
\end{itemize}
  If a requirement is not met, then  $v$  can mark itself $\err$.

\paragraph{Errors in $\{a,q,b\}$-labels in the initial configuration.} We consider the   three labels $a$, $q$, and $b$ in the initial configuration, i.e., the top $s$ vertices in the tree $T_v$. 

We first discuss what are the input labels that we expect to have at these vertices.
We denote these vertices by $v_1, v_2, \ldots, v_s$, where $v_1$ is the root, and $v_{i+1}$ is the child of $v_i$. 

\begin{itemize}
    \item Each of $v_1, v_2, \ldots, v_{s-1}$ has exactly one child, and $v_s$ has exactly two children.
    \item $a = x[i]$ for each $v_i$.
    \item $q = q_0$ for all of $v_1, v_2, \ldots, v_s$.
    \item $b = 1$ for $v_1$, and $b = 0$ for  $v_2, v_3, \ldots, v_s$.
\end{itemize}

To let each vertex $v \in V_A$  decide whether $v = v_i$, we introduce the following $s$ auxiliary symbols: $p_1, p_2, \ldots, p_s$.  The rules for these symbols are as follows.
\begin{itemize}
    \item If $v \in V_A$ is a root, then $v$ is required to mark itself $p_1$.
    \item If the parent of $v \in V_A$ is marked $p_i$ for some $1 \leq i < s$, then $v$ is required to mark itself $p_{i+1}$. 
\end{itemize}
With these symbols, we can now detect errors in $v_i$ by allowing a vertex $v \in V_A$ to mark itself $\err$ whenever $v$  marks itself $p_i$ and $v$ does not satisfy the above requirements for $v_i$.

\paragraph{Errors in $b$-labels in other configurations.}
Let $v \in V_A$ be a vertex that is not within distance $s-1$ to the root. The correct $b$-label at $v$ is determined by the $z$-label at $v$ and the $\{q,a,b\}$-labels in the distance-$\{s-1, s, s+1\}$ ancestors of $v$.

If the $b$-label at $v$ is incorrect, our strategy to detect such an error is to use extra symbols to allow $v$ to send its  $\{b, z\}$-labels to its distance-$s$ ancestor $u$, and then $u$ can decide whether there is indeed an error by inspecting the input labels in $u$ and the neighbors of $u$.

Specifically, we introduce the auxiliary symbols $e_{i, b, z}$, where   $b \in \{0,1\}$, $z \in \{1, 2\}$, and $0 \leq i \leq s$. The rules are the following.
\begin{itemize}
    \item Each $v \in V_A$ can mark itself $e_{0, b, z}$ if the values of $\{b, z\}$ are the same as the $\{b, z\}$-labels at $v$.
    \item For each $1 \leq i \leq s$, $v \in V_A$ can mark itself $e_{i, b, z}$ if $v$ has a child that marks itself  $e_{i-1, b, z}$.
    \item Each $v \in V_A$ can mark itself $\err$ if $v$  marks itself  $e_{s, b, z}$ and  the values of $\{b, z\}$ are inconsistent with the $\{a,q,b\}$-labels at $v$ and the neighbors of $v$ in the sense described above.
\end{itemize}

\paragraph{Errors in $a$-labels in other configurations.}
Let $v \in V_A$ be a vertex that is not within distance $s-1$ to the root. 
Let $u$ be the distance-$s$ ancestor of $v$.
The correct $a$-label at $v$ is determined by the  $z$-label at $v$ and the $\{q,a,b\}$-labels at $u$. An incorrect $a$-label at $v$ can be detected in a way similar to the case of $b$-labels.

\paragraph{Errors in $q$-labels in other configurations.}
As we already force all vertices in a path $P_{C,z}$ to have the same $q$-label, here we only need to focus on the vertices with $b=1$.
Let $v \in V_A$ be a vertex with $b=1$ that is not within distance $s-1$ to the root.
Similar to the case of $b$-labels, the correct $q$-label at $v$ is determined by  the  $z$-label at $v$ and  the $\{q,a,b\}$-labels in the distance-$\{s-1, s, s+1\}$ ancestors of $v$.
Hence an incorrect $q$-label at $v$ can be detected similarly.




\paragraph{Summary.} The number of input labels and output labels needed is polynomial in the description length of $(\MM, x)$. Note that we only use radius-1 constraints in the above construction. Since $\Delta =O(1)$ and $r = O(1)$, the description length of $\mathcal{P}_{\MM,x}$ is polynomial in the description length of $(\MM, x)$, and so this is indeed a polynomial-time reduction.

 As the tree  $T_{\MM, x}$ is independent of the underlying network $G$, the height of $T_{\MM, x}$ can be seen as a constant independent of $n$. If $T_v \neq T_{\MM, x}$, then $v$ is able to mark itself $\bigstar$ by  examining and marking its distance-$O(1)$ neighborhood in $T_v$, and this can be done in $O(1)$ rounds deterministically in the $\LOCAL$ model.


%% file: operations.tex
\section{Basic Graph Operations \label{sect-operations}}
In Section~\ref{sect-operations} we review the graph operations considered in~\cite{ChangP17}.
These operations will be used in the proof of Theorem~\ref{thm-gap-main} in Section~\ref{sect-gap-details}.

\subsection{Graph Surgery}
Let $\GG=(G,\LL)$ be a partially labeled graph, and let $\HH=(H,\LL)$ be a subgraph of $\GG$.
Define the \emph{poles} of $\HH$ as those vertices in $V(H)$ that are adjacent to vertices in $V(G) \setminus V(H)$.
The operation $\replace$  removes $\HH$ and replaces it with some other graph $\HH'$.
\medskip

\begin{description}
\item[$\replace$] Let $S=(v_1, v_2, \ldots,v_p)$ be an ordered list of the poles of $\HH$
and let $S=(v_1', v_2', \ldots,v_p')$ be a designated ordered list of poles in some   graph $\HH'$.
The   graph $\GG' = \replace(\GG,(\HH,S),(\HH',S'))$
is constructed as follows.
Starting with $\GG$, replace $\HH$ with $\HH'$, and replace each edge $\{u,v_i\}$ with $u\in V(G) \setminus  V(H)$ by the new edge $\{u,v_i'\}$.  
\end{description}

If   $S$ and $S'$ are clear, then
we  simply write $\GG' = \replace(\GG,\HH,\HH')$.
Writing $\GG'=(G',\LL')$ and $\HH' =(H',\LL')$, there is a natural 1-1 correspondence between  $V(G) \setminus V(H)$ and $V(G') \setminus V(H')$.

\subsection{An Equivalence Relation on Graphs}
The following two theorems are crucial properties of the equivalence relation $(\HH,S) \simm (\HH',S')$ on partially labeled graphs with a fixed number of poles  defined in~\cite{ChangP17}. 

\begin{theorem} \label{thm:rel-1-copied}
Let $\GG = (G,\LL)$, and let $\HH=(H,\LL)$ be a subgraph $\GG$.
Suppose $\HH'$ is a graph for which $(\HH,S)\simm (\HH',S')$
and let $\GG' = \replace(\GG,(\HH,S),(\HH',S'))$.
We write $\GG' = (G',\LL')$ and $\HH'=(H',\LL')$.
Let $\LL_\diamond$ be a complete labeling of $\GG$ that is locally consistent for all vertices in $H$.
Then there exists a complete labeling $\LL_\diamond'$ of $\GG'$ such that the following conditions are met.
\begin{itemize}
\item  For each $v \in V(G) \setminus V(H)$ and its corresponding $v' \in V(G') \setminus V(H')$, we have $\LL_\diamond(v) = \LL_\diamond'(v')$. Moreover, if $\LL_\diamond$ is locally consistent for $v$, then $\LL_\diamond'$ is locally consistent for $v'$.
\item  $\LL_\diamond'$ is locally consistent for all vertices in $H'$.
\end{itemize}
Moreover, the labeling $\LL_\diamond'$ of $\HH'$ can be computed from $\HH'$, $\HH$, and the labeling $\LL_\diamond$ restricted to $\HH$.
\end{theorem}

\begin{theorem} \label{thm:rel-2-copied}
Let $\GG = (G,\LL)$, and let $\HH=(H,\LL)$ be a subgraph of $\GG$ with poles $S$.
Suppose $\HH'$ is a graph that satisfies $(\HH,S) \simm (\HH',S')$ for some pole list $S'$.
Let $\GG' = \replace(\GG,(\HH,S),(\HH',S'))$.
Designate a set $X \subseteq (V(G) \setminus V(H)) \cup S$ as the poles of $\GG$, listed in some order,
and let $X'$ be the corresponding list of vertices in $\GG'$.
Then we have $(\GG,X) \simm (\GG',X')$.
\end{theorem}

The equivalence relation $\simm$  has a \emph{finite} number of equivalence classes,
for any fixed number $p$ of poles, and for any fixed LCL problem $\mathcal{P}$.


\subsection{A Pumping Lemma for Trees}

A partially labeled tree $\TT$ with one pole $S=(z)$ is called a \emph{rooted} tree with the root $z$. 
Let $\class(\TT)$ denote the equivalence class of   $(\TT, S=\{z\})$ w.r.t.~$\simm$.
By Theorem~\ref{thm:rel-1-copied}, whether a partially labeled rooted tree $\TT$ admits a legal labeling is determined by $\class(\TT)$.
The following lemma is a consequence of Theorem~\ref{thm:rel-2-copied}.


\begin{lemma} \label{lem:replace-rootedtree-copied}
Let $\TT$ be a partially labeled rooted tree, and let $\TT'$ be a rooted subtree of $\TT$,
whose leaves are also leaves of $\TT$.
Let $\TT''$ be another partially labeled rooted tree such that $\class(\TT') = \class(\TT'')$.
Then replacing $\TT'$ with $\TT''$ does not alter the class of $\TT$.
\end{lemma}

A partially labeled tree $\HH$ with two poles $S=(s,t)$ is called a \emph{bipolar} tree.
Let $\HH=(H,\LL)$ be a bipolar tree with poles $S=(s,t)$.
The unique directed path $(s = v_1, v_2, \ldots, v_x = t)$ in $H$ from $s$ to $t$ is called the {\em core path} of $\HH$.
 A bipolar tree $\HH=(H,\LL)$ can be naturally represented as a \emph{sequence of rooted trees} $(\TT_1, \TT_2, \ldots, \TT_x)$, as discussed in Section~\ref{sect:pumping-prelim}.
 Let   $\type(\HH)$ denote the equivalence class of  $(\HH,S=(s,t))$ w.r.t.~$\simm$.
By Theorems~\ref{thm:rel-1-copied} and~\ref{thm:rel-2-copied}, we have the following four lemmas.

\begin{lemma} \label{lem:type2class-copied}
Let $\HH$ be a partially labeled bipolar tree with poles $(s,t)$. Let $\TT$ be $\HH$, but regarded as a  tree rooted at $s$.
Then $\class(\TT)$ is determined by $\type(\HH)$.
If we write $\HH = (\TT_i)_{i \in [x]}$, then $\type(\HH)$ is determined by $\class(\TT_1), \class(\TT_2), \ldots, \class(\TT_{x})$.
\end{lemma}


\begin{lemma} \label{lem:replace-copied}
Let $\GG=(G,\LL)$ be a partially labeled graph, and let $\HH=(H, \LL)$ be a
bipolar subtree of $\GG$ with poles $(s,t)$.
Let $\HH'$ be another partially labeled bipolar tree.
Consider $\GG' = \replace(\GG,\HH,\HH')$.
If $\type(\HH') = \type(\HH)$ and $\GG$ admits a legal labeling $\LL_\diamond$,
then $\GG'$ admits a legal labeling $\LL_\diamond'$ such that
${\LL_\diamond}(v) = \LL_\diamond'(v')$ for each vertex $v \in V(G) \setminus V(H)$ and its corresponding $v' \in V(G') \setminus V(H')$.
\end{lemma}

Furthermore, in Lemma~\ref{lem:replace-copied}, the labeling $\LL_\diamond'$ of vertices $V(H')$ can be computed only using the information of $\HH$, $\HH'$, and the output labeling $\LL_\diamond$ restricted to $\HH$. In other words, the knowledge of anything in  $\GG'$ outside of $\HH'$ is not needed. 

\begin{lemma} \label{lem:replace-2-copied}
Suppose that $\GG = (\TT_i)_{i \in [x]}$ is a partially labeled bipolar tree,
$\HH = (\TT_i, \TT_{i+1}, \ldots, \TT_j)$ is a bipolar subtree of $\GG$,
and $\HH'$ is some other partially labeled bipolar tree with $\type(\HH') = \type(\HH)$.
Then $\GG' = \replace(\GG,\HH,\HH')$ is a partially labeled bipolar tree
and $\type(\GG') = \type(\GG)$.
\end{lemma}

\begin{lemma}  \label{lem:type-copied}
Let $\HH = (\TT_i)_{i \in [x]}$ and
$\HH' = (\TT_i)_{i \in [x+1]}$ be identical to $\HH$ in its first $x$ trees.
Then $\type(\HH')$ is a function of $\type(\HH)$ and $\class(\TT_{x+1})$.
\end{lemma}

Combining Lemma~\ref{lem:type-copied} with the well-known pumping lemma of finite automata, there is a \emph{pumping constant} $\Lpump$ such that the following lemma holds.


\begin{lemma} \label{thm:pump-copied}
Let $\HH = (\TT_1, \TT_2, \ldots, \TT_k)$, with $k \geq \Lpump$.
Regard each $\TT_i$ in the string notation $\HH = (\TT_1, \TT_2, \ldots, \TT_k)$ as a character.
 Then $\HH$ can be decomposed into three substrings $\HH = x \circ y \circ z$ such that (i) $|xy| \leq \Lpump$, (ii) $|y|\geq 1$, and (iii) $\type(x \circ y^j \circ z) = \type(\HH)$ for each non-negative integer $j$.
\end{lemma}


Lemma~\ref{thm:pump-copied} guarantees the existence of the following function $\pump$.

\medskip

\begin{description}
\item[$\pump$]
Let $\HH = (\TT_i)_{i \in [x]}$ be a partially labeled bipolar tree with $x \geq \Lpump$.
The function $\pump(\HH,w)$ produces a partially labeled bipolar tree $\HH' = (\TT_i')_{i \in [x']}$
such that the following conditions are met.
\begin{itemize}
\item The set of rooted trees $\{\TT_1, \TT_2, \ldots, \TT_x\}$ appearing in the tree list of $\HH$ is the same as the set of  rooted trees $\{\TT_1', \TT_2', \ldots, \TT_{x'}'\}$ in the tree list of $\HH'$. 
    \item $\type(\HH) = \type(\HH')$.
\item $x' \in [w,w+\Lpump]$.

\end{itemize}

\end{description}

\subsection{Other Graph Operations}
The following operations $\extend$, $\labelling$, and $\cut$ are taken from~\cite{ChangP17}. The operation $\labelling$ is parameterized by a labeling function $f$, which assigns output labels to the middle part of the input bipolar tree. The operation $\extend$ is parameterized by a number $w$, which indicates the  target length of the output bipolar tree.

\begin{description}
\item[\labelling.]
Let $\HH=(\TT_1, \TT_2,  \ldots, \TT_x)$ be a partially labeled bipolar tree with $x \geq \ell$.
Let $(v_1, v_2, \ldots, v_x)$ be the core path of $\HH$
and $e = \{v_{\lfloor x/2 \rfloor}, v_{\lfloor x/2 \rfloor+1}\}$ be the middle edge of the core path.
It is guaranteed that all vertices in $N^{r-1}(e)$ in $\HH$ are not already assigned output labels.
The partially labeled bipolar tree $\HH' = \labelling(\HH)$ is defined as the result of assigning output labels to vertices in $N^{r-1}(e)$ by the function $f$. 
Note that the neighborhood function is evaluated w.r.t.~$H$.  In particular, the set $N^{r-1}(e)$ contains the vertices $v_{\lfloor x/2 \rfloor-r+1}, \ldots, v_{\floor{x/2} + r}$ of the core path,
and also contains parts of the trees $\TT_{\floor{x/2}-r+1}, \ldots, \TT_{\floor{x/2} + r}$.
\item[\extend.]
Let $\HH=(\TT_1, \TT_2,  \ldots, \TT_x)$ be a partially labeled bipolar tree with $x \in [\ell, 2 \ell]$.
The partially labeled bipolar tree  $\HH' = \extend(\HH)$ is defined as follows.
Consider the decomposition $\HH = \XX \circ \YY \circ \ZZ$, where
$\YY = (\TT_{\lfloor x/2 \rfloor-r+1}, \ldots, \TT_{\lfloor x/2 \rfloor + r})$.
Then $\HH' = \pump(\XX,w) \circ \YY \circ \pump(\ZZ,w)$.
\item[\cut.]
Let $\GG=(G,\LL)$ be a partially labeled graph and $\HH=(H, \LL)$ be a bipolar subtree with poles $(s,t)$.
Suppose that $\HH$ is connected to the rest of $\GG$ via two edges $\{u,s\}$ and $\{v,t\}$.
The partially labeled graph $\GG' = \cut(\GG,\HH)$
is formed by
(i) duplicating $\HH$ and the edges $\{u,s\}$ and $\{v,t\}$ so that $u$ and $v$ are attached to both copies of $\HH$,
(ii) removing the edge that connects $u$ to one copy of $\HH$, and removing the edge from $v$ to the other copy of $\HH$.
\end{description}

The following two lemmas summarize the crucial properties of these graph operations. For notational simplicity,  let $\ext(\HH)$  denote $\extend(\labelling({\HH}))$. 

\begin{lemma} \label{lem:recover-1-copied}
Let $\GG=(G,\LL)$ be a partially labeled graph and $\HH=(H,\LL)$ be a bipolar subtree of $\GG$ with poles $(s,t)$.
Let $\tilde{\HH}$ be another partially labeled bipolar tree with $\type(\tilde{\HH}) = \type(\HH)$
and $\HH' = \ext(\tilde{\HH})$.
If $\GG' = \replace(\GG,\HH,\HH')$ admits a legal labeling $\LL_\diamond'$,
then $\GG$ admits a legal labeling $\LL_\diamond$ such that ${\LL_\diamond}(v) = \LL_\diamond'(v')$ for each vertex $v \in V(G) \setminus V(H)$ and its corresponding vertex $v' \in V(G') \setminus V(H')$.
\end{lemma}

\begin{lemma} \label{lem:recover-2-copied}
Let $\HH = \ext(\tilde{\HH})$ for some partially labeled bipolar tree  $\tilde{\HH}$.
If $\GG' = \cut(\GG,\HH)$ admits a legal labeling $\LL_\diamond'$, then $\GG$ admits a legal labeling $\LL_\diamond$ such that ${\LL_\diamond}(v) = \LL_\diamond'(v')$ for each vertex $v \in V(G) \setminus V(H)$ and its
corresponding vertex $v'$ in $\GG'$.
\end{lemma}

%% file: LCLgap.tex
\section{Existence of the Complexity Gaps}\label{sect-gap-details}

The goal of this section is to provide a complete proof of Theorem~\ref{thm-gap-main}. We begin with defining the sequence of partially labeled virtual graphs $\RRR{1}, \RRC{1}, \RRR{2}, \RRC{2}, \ldots, \RRR{k}$ as in Section~\ref{sect:gap-extend}.
Unless otherwise stated, all graphs in this section are partially labeled.

Here we need one additional operation $\truncate$ whose definition is deferred. This operation satisfies the following properties, and it will play a critical role in establishing the decidability result in Theorem~\ref{thm-gap-main}.
\begin{itemize}
    \item If $\TT$ is a   rooted tree, then $\truncate(\TT)$ is also a  rooted tree of the same class as $\TT$.
    \item If $\HH$ is a   bipolar tree, then $\truncate(\HH)$ is also a   bipolar tree of the same type as $\HH$.
\end{itemize}

\paragraph{Imaginary and real vertices.} Recall from Section~\ref{sect:gap-extend}  that each of   $\RRR{1}, \RRC{1}, \RRR{2}, \RRC{2}, \ldots, \RRR{k}$ consists of {\em imaginary} and {\em real} parts. 
The graph $\RRR{i}$ has $\GR{i}$ as its real part, and the graph $\RRC{i}$ has $\GC{i}$ as its real part. 
Suppose $v$ is a real vertex in  $\RRR{i}$. We let $\TTR{i}(v)$ denote the rooted subtree of $\RRR{i}$ induced by $v$ and all imaginary subtrees stored in $v$, and so the root of $\TTR{i}(v)$  is $v$. The tree $\TTC{i}(v)$ is defined analogously for each real vertex $v$ in $\RRC{i}$.

Let $\Pset_i$ denote the connected components induced by $\VC{i}$ in the graph  $\RRC{i}$. For each $P \in \Pset_i$,   let $\HH_P$ denote the bipolar subtree of  $\RRC{i}$ that includes $P$ and all imaginary subtrees stored in vertices in $P$. If $P = (v_1, v_2, \ldots, v_s)$, then  $\HH_P = (\TTC{i}(v_1), \TTC{i}(v_2), \ldots, \TTC{i}(v_s))$. The core path of $\HH_P$ is $P$. The two poles of $\HH_P$ are $v_1$ and $v_s$.

\paragraph{Construction of $\RRR{i}$.}
For the base case, define $\RRR{1} = \GR{1} = G$. All vertices in $\RRR{1}$ are real, and they are not assigned output labels.

Let $2 \leq i \leq k$. Assume  $\RRC{i-1}$ has been defined. The graph $\RRR{i}$ is defined by the  following procedure. Initialize $\tilde{\GG} \leftarrow \RRC{i-1}$. For each $P \in \Pset_{i-1}$, we compute $\tilde{\HH} = \truncate(\HH_P)$ and $\HH_P^+ = \ext(\tilde{\HH})$,  update $\tilde{\GG} \leftarrow \replace(\tilde{\GG}, \HH_P, \HH_P^+)$, and then $\tilde{\GG} \leftarrow \cut(\tilde{\GG},\HH_P^+)$. 
After these operations, we set $\RRR{i} \leftarrow \tilde{\GG}$.

Intuitively, what these operations do  to the bipolar subtree $\HH_P$ is that we first change it to another one $\tilde{\HH}$ having the same type as $\HH_P$, and then we fix the output labels of its middle part and extend its length using a pumping lemma, and then we duplicate the resulting bipolar subtree using $\cut$. The purpose of these operations is to deal with the issues  discussed in Section~\ref{sect:gap-extend}.

\paragraph{Construction of $\RRC{i}$.}
Let $1 \leq i \leq k-1$. Assume  $\RRR{i}$ has been defined.  The graph $\RRC{i}$ is defined by the following procedure.
 Initialize $\tilde{\GG} \leftarrow \RRR{i}$ after the removal of each connected component that does not contain a vertex in $V(\GC{i})$. For each $v \in V(\GC{i})$ in the graph $\RRR{i}$, do the following. Let $\TT_v$ be the subtree induced by $v$ and all vertices in $V(\RRR{i}) \setminus V(\GC{i})$ that are reachable to $v$ via vertices in $V(\RRR{i}) \setminus V(\GC{i})$. Here $\TT_v$ is interpreted as a tree rooted at $v$, and $\GC{i}$ is interpreted as a subgraph of $\RRR{i}$.  Let $\TT_v^+ = \truncate(\TT_v)$, and 
  update $\tilde{\GG} \leftarrow \replace(\tilde{\GG}, \TT_v, \TT_v^+)$.
 After these operations, we set  $\RRC{i} \leftarrow \tilde{\GG}$.

 \paragraph{An overview of the proof.}
 Our plan is to show that any $O(n^{1/k})$-round solvable LCL problem $\mathcal{P}$ can be solved in the following canonical way.
 First, construct an  ($\gamma, \ell$)-decomposition, with
 \begin{align*}
     \gamma &= n^{1/k} (\ell /2)^{1 - 1/k} = \Theta(n^{1/k}) \ \  \text{ and}\\
     \ell &= 2(r+\Lpump) = \Theta(1).
 \end{align*}
 This decomposes  the vertex set $V$ into
$
 (\VR{1}, \VC{1}, \VR{2}, \VC{2},  \ldots,   \VR{k})$.
 
Then we construct the  virtual graphs $\RRR{1}, \RRC{1}, \RRR{2}, \RRC{2}, \ldots, \RRR{k}$. Note that the partial labeling in these graphs depend on the labeling function $f$ inside the operation $\labelling$. We will later see that if  $\mathcal{P}$ can be solved in $o(n^{1/(k-1)})$ rounds, then there is a labeling function $f$ guaranteeing that these partial labelings can be completed into a complete legal labeling. 

Remember the real part of $\RRR{k}$ is $\VR{k}$, and so a connected component of $\RRR{k}$ is stored by vertices in a connected component of $\VR{k}$. Since each connected component of $\VR{k}$ has size at most $O(n^{1/k})$, a  complete legal labeling of $\RRR{k}$ can be computed in $O(n^{1/k})$ rounds.

In view of Lemmas~\ref{lem:replace-copied},~\ref{lem:recover-1-copied}, and~\ref{lem:recover-2-copied}  and the recursive construction of $\RRR{i}$ and  $\RRC{i}$, given a legal labeling of $\RRR{k}$, it is possible to obtain a complete legal labeling of $\RRC{k-1}, \RRR{k-1}, \ldots, \RRR{1}=G$, and this can be implemented in the deterministic $\LOCAL$ model in $O(n^{1/k})$ rounds. The round complexity analysis relies on the fact that each connected component of $\VR{i}$ has diameter $O(n^{1/k})$.

To summarize, the goal is to show that  if $\mathcal{P}$ can be solved in $o(n^{1/(k-1)})$ rounds by a randomized $\LOCAL$ algorithm, then there exists a {\em feasible} labeling function $f$ which implies the existence of an $O(n^{1/k})$-round deterministic $\LOCAL$ algorithm for $\mathcal{P}$. Moreover, we will show that whether a feasible function $f$ exists is \emph{decidable}.

\paragraph{Organization.}
 In Section~\ref{sect-feasible-function-}, we formally define what makes a labeling function $f$ feasible. We will show that (i)  if $\mathcal{P}$ can be solved in randomized $o(n^{1/(k-1)})$ rounds, then there exists a feasible function, and (ii) the existence of a feasible function is decidable.
In Section~\ref{sect-alg-details} we give a  detailed description of how we can obtain a complete legal labeling of $\RRC{k-1}, \RRR{k-1}, \ldots, \RRR{1}=G$ from any given legal labeling of $\RRR{k}$ in $O(n^{1/k})$ rounds.

\subsection{Feasible Function}\label{sect-feasible-function-}

We will define a hierarchy of subsets of rooted trees $\Tset_1, \Tset_2, \ldots, \Tset_k$,  rooted trees $\Tset_1^+, \Tset_2^+,\ldots, \Tset_k^+$, bipolar trees $\Hset_1,  \Hset_2,\ldots, \Hset_{k-1}$, and bipolar trees $\Hset_1^+, \Hset_2^+, \ldots, \Hset_{k-1}^+$.  All these trees are  partially labeled. The formal definition of these sets also includes the definition of the $\truncate$ operation.

\begin{description}
\item[$\Tset$ Sets:]
The set $\Tset_1$ consists of unlabeled rooted trees of maximum degree $\Delta$.
For $1 < i \leq k$, we have $\TT \in \Tset_i$ if $\TT$ can be constructed as follows. 
Start with any unlabeled rooted trees $\TT'$ of maximum degree $\Delta$. For each vertex $v$ in  $\TT'$, attach to $v$ at most $\Delta - \deg(v)$ bipolar trees from $\Hset_{i-1}^+$ by adding an edge linking $v$ with  a pole of the bipolar tree.

\item[$\Tset^+$ Sets:]
The set $\Tset_i^+$ is defined by the following construction procedure.
If $i = 1$, we initially set $\Tset_i^+ \leftarrow \emptyset$. If $i > 1$, we initially set $\Tset_i^+ \leftarrow \Tset_{i-1}^+$.
Then, we process each tree $\TT \in \Tset_i$ one by one, prioritizing trees with smaller height. When we process  $\TT$, we first check if there is already a tree $\TT' \in \Tset_i^+$ such that $\class(\TT) = \class(\TT')$. If so, then we define $\truncate(\TT) = \TT'$. Otherwise, we add $\TT$ to the set $\Tset_i^+$, and define $\truncate(\TT) = \TT$.

\item[$\Hset$ Sets:]
The set ${\Hset}_i$ contains all  bipolar trees $\HH = (\TT_j)_{j \in [x]}$
such that $x\in [\ell,2\ell]$,
and for each $j \in [x]$, the tree $\TT_j$ is constructed as follows.
Start with the root vertex $z$, and attach to $z$ at most $\Delta - 2$ trees from ${\Tset}_i^+$ whose root has maximum degree at most $\Delta - 1$.

\item[$\Hset^+$ Sets:]
The set $\Hset^+$ is defined by the following construction procedure.
If $i = 1$, we initially set $\Hset_i^+ \leftarrow \emptyset$. If $i > 1$, we initially set $\Hset_i^+ \leftarrow \Hset_{i-1}^+$.
Then, we process each tree $\HH \in \Hset_i$ one by one in any order.  When we process  $\HH$, we first check if there is already a
 tree $\HH' \in \Hset_i$ such that
$\type({\HH'}) = \type(\HH)$ and $\ext({\HH'}) \in {\Hset}_i^+$.
If so, we define $\truncate(\HH) = \HH'$.
Otherwise, we add $\ext({\HH})$ to the set ${\Hset}_i^+$, and define
 $\truncate(\HH) =\HH$.
\end{description}

\paragraph{Relation to the Virtual Graphs.} Before we proceed, we discuss how the above sets relate to the sequence of virtual graphs
$\RRR{1}, \RRC{1}, \RRR{2}, \RRC{2}, \ldots, \RRR{k}$.
The following statements are straightforward to verify by inspection of the definition of these sets. For $i \in [1, k-1]$, the set $\Tset_i$ is a superset of all rooted trees that might appear as a connected component in $\RRR{i} \setminus \GC{i}$. In particular, for $i=k$, the set $\Tset_k$ represents the a superset  of all possibles connected components that might appear in $\RRR{k}$.
The set $\Tset_i^+$ is a superset of all possible imaginary trees whose parent is a real vertex $v \in V(\GC{i})$ in $\RRC{i}$.
The set $\Hset_i$  is a superset of all possible $\HH_P$, for $P \in \Pset_i$.
The set $\Hset_i^+$ then represents a superset of all possible $\HH_P^+$, for $P \in \Pset_i$.

\paragraph{Feasible Function.} The correctness of our approach   relies on the fact that the partial labeling of $\RRR{k}$ can be completed into a  complete legal labeling. To show that this is true, it suffices to show that each member of $\Tset_k$ admits a legal labeling.
In particular, whether this is true only depends on the set of all classes in $\Tset_k$.

The set $\Tset_k$ is well-defined so long as the labeling function $f$ inside $\labelling$ and the parameter $w$ inside $\extend$ are fixed. Remember that $\ell$ is always defined as $2(r+\Lpump)$. The choice of $w$ represents the target length of the pumping lemma, and it does not affect the set of classes in $\Tset_k$. Hence the only adjustable parameter that matters is $f$. In view of the above, we say that $f$ is {\em feasible} if it leads to a set $\Tset_k$ where all members admit a legal labeling.

\paragraph{Properties of the Sets.}
We  analyze the properties of the four sequences of sets. We aim to prove the following statements.
\begin{itemize}
    \item The set of classes in  $\Tset_k$ is invariant of the parameter $w$. This implies that  feasible function is well-defined, see Lemma~\ref{lem:w-converge-analogous}.
    \item Given an LCL problem $\mathcal{P}$, whether a feasible function $f$ exists is decidable, see Lemma~\ref{lem-feasible-func-decide}.
    \item If there is a randomized $\LOCAL$ algorithm solving $\mathcal{P}$ in $o(n^{1/(k-1)})$ rounds, then there exists a feasible function, see Lemma~\ref{lem:func-exist-analogous}.
\end{itemize}

If $\Tset$ is a set of rooted trees, we define $\class(\Tset) = \{\class(\TT) \;|\; \TT\in \Tset\}$. If $\Hset$ is a set of bipolar trees, we define
  $\type(\Hset) = \{\type(\HH) \;|\; \HH\in\Hset\}$. The {\em height} of a rooted tree is defined as the maximum distance between a vertex and the root.

\begin{lemma}\label{lem:w-converge-analogous}
For each $i \in [k]$,
$\class(\Tset_i)$ does not depend on the parameter $w$ used in the operation $\extend$.
\end{lemma}
\begin{proof}
By its definition, $\class(\Tset_1)$ is invariant of $w$.
Suppose, by induction, that $\class(\Tset_{i-1})$ is invariant of $w$.
Then we claim that 
$\type(\Hset_{i-1})$, $\type(\Hset_{i-1}^+)$,  and $\class(\Tset_{i})$ are also invariant of $w$.

The fact that $\type(\Hset_{i-1})$ is  invariant of $w$ is due to Lemma~\ref{lem:replace-rootedtree-copied}.
The fact that $\type(\Hset_{i-1}^+)$ is  invariant of $w$ is due to 
the fact that
the type of $\HH' = \extend(\HH)$ is invariant over all choices of the parameter $w$.
Given that $\type(\Hset_{i-1}^+)$ is invariant of $w$, we can infer that  $\class(\Tset_{i})$ is invariant of $w$ in view of Lemma~\ref{lem:replace-rootedtree-copied} and Lemma~\ref{lem:type2class-copied}.
\end{proof}

We define $\Tset_i^{j}$ as follows.
Start with any unlabeled rooted trees $\TT'$ of maximum degree $\Delta$ and height at most $j$. For each vertex $v$ in  $\TT'$, attach to $v$ at most $\Delta - \deg(v)$ bipolar trees from $\Hset_{i-1}$ by adding an edge linking $v$ with one pole of the bipolar tree. For the special case of $i=1$, define $\Tset_1^{j}$ as the set of rooted trees of maximum degree $\Delta$ and height at most $j$.
We have 
\[\Tset_i^{0} \subseteq \Tset_i^{1} \subseteq 
\Tset_i^{2} \subseteq \cdots \subseteq  \Tset_i.\]

\begin{lemma}\label{lem-tree-height}
There exists a universal constant $h_0$ depending only on the description of the LCL $\mathcal{P}$ such that $\Tset_i^+ \subseteq \Tset_i^{h_0}$.
\end{lemma}
\begin{proof}
We claim that $\class(\Tset_i^{j+1})$ depends only on $\class(\Tset_i^{j})$. That is, if $\class(\Tset_i^{h+1}) = \class(\Tset_i^{h})$, then we have $\class(\Tset_i^{h'}) = \class(\Tset_i^{h})$ for all $h' \geq h$, which means that $\class(\Tset_i^{h}) = \class(\Tset_i)$, and so  all trees in $\Tset_i^+$ must come from $\Tset_i^{h}$. Since the total number of classes is a constant,  there exists a  constant $h_0$ such that $\class(\Tset_i^{h_0+1}) = \class(\Tset_i^{h_0})$.

For the rest of the proof, we show that $\class(\Tset_i^{j+1})$ depends only on $\class(\Tset_i^{j})$. The set $\class(\Tset_i^{j+1})$ can be constructed alternatively as follows. Start with the root vertex $z$. Attach to $z$ at most $\Delta$ trees with maximum degree at most $\Delta-1$ from the set $\class(\Tset_i^{j})$.
The set $\class(\Tset_i^{j+1})$ is exactly the set of all trees that can be constructed using the above procedure. In view of Lemma~\ref{lem:replace-rootedtree-copied}, once the set of all classes of the trees of degree at most  $\Delta-1$ in $\Tset_i^{j}$ is fixed, then we can infer  $\class(\Tset_i^{j+1})$.
\end{proof}

\begin{lemma}\label{lem-feasible-func-decide}
Given an LCL problem $\mathcal{P}$, whether a feasible function $f$ exists is decidable.
\end{lemma}
\begin{proof}
By its definition, $f$ is feasible if all members in $\Tset_k$ admit a legal labeling.
Whether a rooted tree admits a legal labeling depends only on its class,  so we can consider the set $\Tset_k^{+}$ instead of $\Tset_k$.

Given a labeling function $f$, we show that it is decidable whether $f$ is feasible.
The only part of the construction of $\Tset_k^{+}$  that involve infinite number of steps is the construction of the sets $\Tset_1,  \Tset_2, \ldots \Tset_k$. However, by Lemma~\ref{lem-tree-height}, we can simply replace each $\Tset_i$ with $\Tset_i^{h_0}$, which can be constructed using a bounded number of steps.

Next, we argue that the number of labeling function $f$ needed to consider is also bounded. In view of the way we define $\Hset_i^{+}$, we only compute $\labelling(\HH)$ for at most a constant number of $\HH$, and this number is upper bounded by the number of types. Therefore, the   number of labeling function $f$ needed to consider is also bounded.
\end{proof}

For the case the locality radius $r = O(1)$ is a constant independent of the  description length $N$ of $\mathcal{P}$, the algorithm implicitly in Lemma~\ref{lem-feasible-func-decide} runs in time $2^{2^{N^{O(1)}}}$.
The reason is that the number of classes and the number of types is $2^{N^{O(1)}}$ when the locality radius $r$ is a constant, see~\cite[Section 3.7]{ChangP17} for the calculation. Remember that in our construction of the four sequences of sets, the number of times we apply the labeling function $f$ is at most the number of types.
For a given bipolar tree $\HH$, the number of ways of labeling its middle part in $\labelling$ is at most $|\LabelOut|^{\Delta^{r-1}+1} < N$. The number of $\HH$ we apply the labeling function is at most the number of types, and hence the total number of possibilities we need to try is at most $N^{2^{N^{O(1)}}} = 2^{2^{N^{O(1)}}}$.
Note that we do not need to explicitly construct these four sequences of sets, and it suffices to maintain the set of types and the set of classes for these sets.

Lemma~\ref{lem-final-tree-size} is a straightforward consequence of Lemma~\ref{lem-tree-height} and the definition of the four sequences of sets.

\begin{lemma}\label{lem-final-tree-size}
Assuming that $k$, $r$, $\Delta$, $\Lpump$ are constants, the number of vertices in a tree in  $\Tset_k^{+}$ is $O(w^{k-1})$.
\end{lemma}

\begin{lemma}\label{lem:func-exist-analogous}
Suppose that there exists a randomized $\LOCAL$ algorithm $\mathcal{A}$ that solves $\mathcal{P}$ in $o(n^{1/(k-1)})$ rounds
on $n$-vertex bounded degree trees, with local probability of failure at most $1/n$.
Then there exists a feasible function $f$.
\end{lemma}
\begin{proof}
Define $\beta = |\LabelOut|^{\Delta^{r}}$ to be an upper bound on the number of distinct output labelings of $N^{r-1}(e)$, where $e$ is any edge in any graph of maximum degree $\Delta$.
Define $N$ as the maximum number of vertices of a tree in ${\Tset}_{k}^+$
over {\em all} choices of labeling function $f$.
As $\Delta, r, k, \Lpump$ are all constants, we have $N = O(w^{k-1})$, see Lemma~\ref{lem-final-tree-size}.
Define $t$ to be the runtime of $\mathcal{A}$ on a $(\beta N + 1)$-vertex tree.
Note that $t$ depends on $N$, which depends on $w$.
The rest of the proof follows from the proof of~\cite[Lemma~15]{ChangP17}.
The high-level idea is that we select $w$ to be sufficiently large such that
$w \geq 4(r+t)$. Such a $w$ exists since $\mathcal{A}$ runs in $o(n^{1/(k-1)})$ rounds on an $n$-vertex graph.
The function $f$ will be chosen as the {\em most probable} outcome of simulating $\mathcal{A}$.
Even though $\mathcal{A}$ is a randomized algorithm,  the function $f$ is defined {\em deterministically}. The argument in the proof of~\cite[Lemma~15]{ChangP17} shows that such function $f$ is feasible.
\end{proof}

\subsection{Details of the \texorpdfstring{$O(n^{1/k})$}{O(n to the power of 1/k)}-Round Algorithm}\label{sect-alg-details}

The goal of this section is to prove the following lemma.  This lemma, together with 
Lemma~\ref{lem-feasible-func-decide} and Lemma~\ref{lem:func-exist-analogous}, implies Theorem~\ref{thm-gap-main}.

\begin{lemma}\label{lem:lognalg-analogous}
Let $\mathcal{P}$ be any LCL problem on trees with $\Delta = O(1)$.
Given a feasible function $f$,
the LCL problem $\mathcal{P}$ can be solved in $O(n^{1/k})$ rounds in the deterministic $\LOCAL$ model.
\end{lemma}


\begin{proof}
Suppose $f$ is a feasible function, and we will use $f$ to design an $O(n^{1/k})$-round deterministic algorithm for $\mathcal{P}$.
Recall that the first step of the $O(n^{1/k})$-round algorithm is to 
 construct an  ($\gamma, \ell$)-decomposition, with $\gamma = n^{1/k}(2\ell)^{1-1/k} = \Theta(n^{1/k})$ and $\ell = 2(r+\Lpump) = \Theta(1)$. This decomposes  the vertex set $V$ into
 \[
 (\VR{1}, \VC{1}, \ldots, \VR{k-1}, \VC{k-1}, \VR{k}).
 \]
 
Then we construct the  virtual graphs $\RRR{1}, \RRC{1}, \RRR{2}, \RRC{2}, \ldots, \RRR{k}$. Since $f$ is feasible,  the  graph $\RRR{k}$ admits a legal labeling, and this legal labeling can be found in  $O(n^{1/k})$ rounds, as each connected component of $\VR{k}$, i.e., the real part of $\RRR{k}$, has diameter  $O(n^{1/k})$.
For the rest of the proof, we show that (i) given a legal labeling of $\RRR{i}$, we can construct a legal labeling of $\RRC{i-1}$ in $O(1)$ rounds, and (ii) given a legal labeling of $\RRC{i}$, we can construct a legal labeling of $\RRR{i}$ in $O(n^{1/k})$ rounds. Applying these constructions iteratively, we can obtain a legal labeling of $G = \RRR{1}$ in $O(n^{1/k})$ rounds, as $k = O(1)$.

\paragraph{From Legal Labeling of $\RRR{i}$ to Legal Labeling of  $\RRC{i-1}$.}
Let $1 < i \leq k$.  Recall that the graph $\RRR{i}$ is defined by the  following procedure. 
\begin{enumerate}
    \item Initialize $\tilde{\GG} \leftarrow \RRC{i-1}$. 
    \item For each $P \in \Pset_{i-1}$, we compute $\tilde{\HH} = \truncate(\HH_P)$ and $\HH_P^+ = \ext(\tilde{\HH})$,  update $\tilde{\GG} \leftarrow \replace(\tilde{\GG}, \HH_P, \HH_P^+)$, and then $\tilde{\GG} \leftarrow \cut(\tilde{\GG},\HH_P^+)$. 
    \item Set $\RRR{i} \leftarrow \tilde{\GG}$.
\end{enumerate}
Let $\GG_a$ be the graph $\tilde{\GG}$ in the middle of the above procedure when we are about to process a path $P \in \Pset_{i-1}$. Define 
$\GG_b = \replace(\tilde{\GG}, \HH_P, \tilde{\HH})$, where $\tilde{\HH} = \truncate(\HH_P)$, and 
$\GG_c = \replace(\GG_b, \tilde{\HH}, \HH_P^+)$,
and $\GG_d = \cut(\tilde{\GG},\HH_P^+)$.
Suppose we have a legal labeling $\LL$ of $\GG_d$. We show that there is a legal labeling $\LL'$ of $\GG_a$  such that the labeling $\LL'$  of $\GG_a$ is identical the the labeling $\LL$  of $\GG_d$ everywhere except $\HH_P$.

In view of  Lemmas~\ref{lem:recover-2-copied},~\ref{lem:recover-1-copied}, and~\ref{lem:replace-copied}, given a legal labeling of $\GG_d$, we can obtain a legal labeling of $\GG_c$ (using Lemma~\ref{lem:recover-2-copied}), $\GG_b$ (using Lemma~\ref{lem:recover-1-copied}), and $\GG_a$ (using Lemma~\ref{lem:replace-copied}). Remember that the operation $\tilde{\HH} = \truncate(\HH_P)$ guarantees that $\tilde{\HH}$ and $\HH_P$ have the same type. The legal labeling of $\GG_a$ is the same as the legal labeling of $\GG_d$ everywhere except $\HH_P$. 

In the virtual graph $\RRC{i-1}$, the bipolar subtree $\HH_P$ is stored in the path $P$ in the underlying graph $G$. To construct a legal labeling of $\RRC{i-1}$ from a legal labeling of $\RRR{i}$ we can apply the above procedure in parallel for each  $P \in \Pset_{i-1}$, since the computation for the labeling of $\HH_P$  is independent for each  $P \in \Pset_{i-1}$. The round complexity is $O(1)$, since the length of each $P \in \Pset_{i-1}$ is $O(1)$.

\paragraph{From Legal Labeling of $\RRC{i}$ to Legal Labeling of  $\RRR{i}$.}
Let $1 \leq i < k$. The graph $\RRC{i}$ is defined by the following procedure.
\begin{enumerate}
\item Initialize $\tilde{\GG} \leftarrow \RRR{i}$ after the removal of each connected component that does not contain a vertex in $V(\GC{i})$. \item For each $v \in V(\GC{i})$ in the graph $\RRR{i}$, do the following. Let $\TT_v$ be the subtree induced by $v$ and all vertices in $V(\RRR{i}) \setminus V(\GC{i})$ that are reachable to $v$ via vertices in $V(\RRR{i}) \setminus V(\GC{i})$. Here $\TT_v$ is interpreted as a tree rooted at $v$, and $\GC{i}$ is interpreted as a subgraph of $\RRR{i}$.  Let $\TT_v^+ = \truncate(\TT_v)$, and 
  update $\tilde{\GG} \leftarrow \replace(\tilde{\GG}, \TT_v, \TT_v^+)$.
  \item
 After these operations, we set  $\RRC{i} \leftarrow \tilde{\GG}$.
 \end{enumerate}
 
 We write $\GG'$ to denote the intermediate graph $\tilde{G}$ at the beginning of the above Step~2. Given a legal labeling of $\GG'$, to compute a legal labeling of $\RRR{i}$, it suffices to show that each connected component of  $\RRR{i}$ that does not contain a vertex in $V(\GC{i})$ admits a legal labeling. This is true, since $f$ is feasible and any such component is a tree in $\Tset_i$. Remember that each connected component of $\RRR{i}$ is stored in a connected component of $\VR{i}$, which is of diameter  $O(n^{1/k})$. Therefore, given a legal labeling of $\GG'$, we can obtain a legal labeling of $\RRR{i}$   in $O(n^{1/k})$ rounds.
 
 Now we show how to obtain a legal labeling of $\GG'$  from a given legal labeling of $\RRC{i}$. In view of how $\RRC{i}$ is constructed from $\GG'$ and Lemma~\ref{lem:replace-copied}, the desired legal labeling of $\GG'$  can be obtained by letting each  $\TT_v$ locally compute a labeling from the given labeling of $\TT_v^+$. Remember that $\TT_v^+ = \truncate(\TT_v)$ has the same class of $\TT_v$. The real part of $\TT_v$ in $\RRR{i}$ is connected and has diameter $O(n^{1/k})$, and so the computation takes $O(n^{1/k})$ rounds.
 \end{proof}